\documentclass[a4paper,UKenglish,cleveref,autoref,thm-restate]{lipics-v2021}

\usepackage{framed,latexsym,amssymb,amsmath,amsthm,float,cite,complexity,setspace,color,graphicx,lineno,lmodern}
\usepackage[shortlabels]{enumitem}
\usepackage[ddmmyyyy]{datetime}
\usepackage[shortend, noend]{algorithm2e}
\usepackage{tikz,pgfplots}
\pgfplotsset{compat=1.17}

\tikzset{
	dot/.style = {circle, fill, minimum size=#1,
		inner sep=0pt, outer sep=0pt},
	dot/.default = 5pt
}
\usetikzlibrary{math}

\usepackage{thmtools,thm-restate}
\usepackage{regexpatch}
\makeatletter
\xpatchcmd\thmt@restatable{%
\csname #2\@xa\endcsname\ifx\@nx#1\@nx\else[{#1}]\fi
}{%
\ifthmt@thisistheone
\csname #2\@xa\endcsname\ifx\@nx#1\@nx\else[{#1}]\fi
\else
\csname #2\@xa\endcsname[{restated}]
\fi}{}{}
\makeatother

\floatstyle{boxed}
\newfloat{problem}{tbhH}{prb}

\newcommand{\Col}{\textup{\textsc{Colouring}}}
\newcommand{\kCol}[1]{\textup{\textsc{#1-Colouring}}}
\newcommand{\ProbeCol}{\textup{\textsc{Colouring}}}
\newcommand{\ProbekCol}[1]{\textup{\textsc{#1-Colouring}}}
\newcommand{\PreExt}{\textup{\textsc{1-Precolouring Extension}}}
\newcommand{\ExactThreeCover}{\textsc{Exact 3-Cover}}

\newtheorem{open}[theorem]{Problem}

\title{Colouring Probe $H$-Free Graphs} 
\titlerunning{Colouring Probe $H$-Free Graphs}

\author{Dani\"el Paulusma}{Department of Computer Science, Durham University, Durham, United Kingdom}{daniel.paulusma@durham.ac.uk}{0000-0001-5945-9287}{}
\author{Johannes Rauch}{Institute of Optimization and Operations Research, Ulm University, Ulm, Germany}{johannes.rauch@uni-ulm.de}{0000-0002-6925-8830}{Supported by the German Academic Scholarship Foundation (Studienstiftung des Deutschen Volkes).}
\author{Erik Jan van Leeuwen}{Department of Information and Computing Sciences, Utrecht University, Utrecht, The Netherlands}{e.j.vanleeuwen@uu.nl}{0000-0001-5240-7257}{}

\authorrunning{D. Paulusma, J. Rauch, and E.J. van Leeuwen.}

\Copyright{Dani\"el Paulusma, Johannes Rauch and Erik Jan van Leeuwen}



\category{}
\relatedversion{}
\acknowledgements{}

\nolinenumbers
\hideLIPIcs  

\EventEditors{John Q. Open and Joan R. Access}
\EventNoEds{2}
\EventLongTitle{42nd Conference on Very Important Topics (CVIT 2016)}
\EventShortTitle{CVIT 2016}
\EventAcronym{CVIT}
\EventYear{2016}
\EventDate{December 24--27, 2016}
\EventLocation{Little Whinging, United Kingdom}
\EventLogo{}
\SeriesVolume{42}
\ArticleNo{23}

\begin{document}
\maketitle

\begin{abstract}
The \NP-complete problems \Col{} and \kCol{$k$} $(k\geq 3$) are well studied on $H$-free graphs, i.e., graphs that do not contain some fixed graph $H$ as an induced subgraph. We research to what extent the known polynomial-time algorithms for $H$-free graphs can be generalized if we only know some of the edges of the input graph. We do this by considering the classical probe graph model introduced in the early nineties. For a graph~$H$, a partitioned probe $H$-free graph $(G,P,N)$ consists of a graph $G=(V,E)$, together with a set $P\subseteq V$ of probes and an independent set $N=V\setminus P$ of non-probes, such that $G+F$ is $H$-free for some edge set $F\subseteq \binom{N}{2}$. We show the following:
\begin{itemize}
\item  We fully classify  \Col{} on partitioned probe $H$-free graphs and show that the obtained complexity dichotomy differs from the known dichotomy of \Col{} for $H$-free graphs.\\[-5pt]
\item We fully classify \kCol{$3$} on partitioned probe $P_t$-free graphs: we prove polynomial-time solvability for $t\leq 5$ and \NP-completeness for $t\geq 6$. In contrast, \kCol{$3$} on $P_t$-free graphs is known to be polynomial-time solvable for $t\leq 7$ and quasi-polynomial-time solvable for $t\geq 8$. 
\end{itemize}
Our main result is our polynomial-time algorithm for {\sc $3$-Colouring} on partitioned $P_5$-free graphs. For this result, and also for all our other polynomial-time results, we do not need to know the edge set $F$; we only need to know its existence. 
Moreover, the class of probe $P_5$-free graphs includes not only paths of arbitrary length but even all bipartite graphs
and is much richer than the class of $P_5$-free graphs. The latter is also evidenced by the fact that there exist graph problems, such as {\sc Matching Cut}, that are known to be polynomial-time solvable for $P_5$-free graphs but \NP-complete for partitioned probe $P_5$-free graphs.  In particular, unlike the class of $3$-colourable $P_5$-free graphs, the class of $3$-colourable probe $P_5$-free graphs has unbounded mim-width.  Hence, our polynomial-time result for $3$-{\sc Colouring} for probe $P_5$-free graphs suggests that there may be another, deeper overarching reason why {\sc $3$-Colouring} is polynomial-time solvable for $P_5$-free graphs. 
\end{abstract}

\section{Introduction}\label{sec:intro}
\Col{} is a classical graph problem. Given a graph~$G$ and a positive integer $k$, it asks whether it is possible to colour the vertices of $G$ with $k$ colours such that any two adjacent vertices receive different colours. The variant where $k$ is fixed beforehand, and not part of the input anymore, is known as \kCol{$k$}.
It is well known that \kCol{3}, and thus \Col{}, are \NP-complete problems~\cite{garey1979computers}.
This led to a rich body of literature
that tries to understand what graph structure causes the computational hardness in \Col{}. 
In our paper we extend this body of work by researching the computational complexity of \Col{} and \kCol{$k$} on classes of graphs that generalize the well-known $H$-free graphs (a graph $G$ is {\it $H$-free} if $G$ does not contain $H$ as an induced subgraph) but for which we do not know all the edges.
Before discussing our model of incomplete information, we first briefly survey some relevant 
known results for \Col{} and \kCol{$k$}.

\vspace*{-0.3cm}
\subparagraph*{${\mathbf{H}}$-Free Graphs} 
Kr{\'{a}}l et al.~\cite{kral2001complexity} showed that if $H$ is a (not necessarily proper) induced subgraph of $P_4$ or $P_3 + P_1$, 
where $P_t$ denotes the path on $t$ vertices, 
then \Col{} on $H$-free graphs is solvable in polynomial time; otherwise, it is \NP-complete.
For \kCol{$k$}, the complexity status on $H$-free graphs has not been resolved yet. 
For every $k\geq 3$, {\sc $k$-Colouring} for $H$-free graphs is \NP-complete if $H$ has a cycle~\cite{EHK98} or an induced claw~\cite{holyer1981thenpcompleteness,LG83}. However, the remaining case where $H$ is a linear forest (disjoint union of paths) has not been settled yet. For $P_t$-free graphs, the cases $k\leq 2$, $t\geq 1$ (trivial), $k\geq 3$, $t\leq 5$~\cite{HKLSS10}, $k=3$, $6\leq t\leq 7$~\cite{BCMSSZ18} and $k=4$, $t=6$~\cite{chudnovsky2024fourcolouring1,chudnovsky2024fourcolouring2} are polynomial-time solvable and the cases  $k=4$, $t\geq 7$~\cite{Hu16} and  $k\geq 5$, $t\geq 6$~\cite{Hu16} are \NP-complete. The cases $k=3$ and $t\geq 8$ are still open, despite some evidence that these cases are polynomial-time solvable due to a quasi-polynomial-time algorithm~\cite{PPR21}.
We refer to the survey~\cite{GJPS17} and some later articles~\cite{CHS24,CHSZ21,HLS22,KMMNPS20} for 
results on \kCol{$k$} for $H$-free graphs if $H$ is a disconnected linear forest, 
and to~\cite{JKMM}  for the most recent results on \Col\ for $(H_1,H_2)$-free graphs, for which also still many cases remain open.

\vspace*{-0.3cm}
\subparagraph*{Probe ${\mathbf{H}}$-Free Graphs}
In this article, we aim to further our understanding of the complexity of \Col{} and \kCol{$k$} by studying \emph{probe graphs}. 
Probe graphs are used to model graphs for which the global structure is known (e.g.\ $H$-freeness). However, we only know the complete set of neighbours for {\it some} vertices of a probe graph~$G$. These vertices are the {\it probes} of $G$. The other vertices are the {\it non-probes} of $G$ and form an independent set in $G$, as we do not know which of them are adjacent to each other. We only know that there exists a ``certifying'' set $F$ of edges on the non-probes such that $G+F$ exhibits the global structure (e.g. being $H$-free). In particular, the subgraph of $G$ induced by the set of probes already has this global structure (e.g. is $H$-free).
The notion of probe graphs was introduced by Zhang et al.~\cite{zhang1994analgorithm} in genome research 
 to make a genome mapping process more efficient.

Formally, for a graph class ${\mathcal G}$, the class ${\cal G}_p$ consists of all graphs $G$ that can be modified into a graph from ${\cal G}$ by adding some edges between an independent set $N$ of $G$.
If for a graph in~${\mathcal G}_p$, the sets $P$ and $N=V\setminus P$ are given, then we speak of a {\it partitioned} probe graph.
Hence, a {\it partitioned probe $H$-free} graph $(G,P,N)$ consists of a graph $G=(V,E)$, together with a set $P\subseteq V$ of probes and an independent set $N=V\setminus P$ of non-probes, such that $G+F$ is $H$-free for some edge set $F\subseteq \binom{N}{2}$. We note that an $H$-free graph itself is also a (partitioned) probe $H$-free graph, namely with $P=V$ and $N=\emptyset$. Hence, for every graph $H$, the class of (partitioned) probe $H$-free graphs contains the class of $H$-free graphs. Consequently, any \NP-completeness results for $H$-free graphs immediately carry over to partitioned probe $H$-free graphs. However, it also leads to the following research question:

\medskip
\noindent
\textit{If an \NP-complete problem $\Pi$ is polynomial-time solvable on $H$-free graphs for some graph~$H$, is $\Pi$ also polynomial-time solvable on (partitioned) probe $H$-free graphs?}

\vspace*{-0.5cm}
\subparagraph{Our Focus}
To investigate our research question, we consider \Col{} and \kCol{$k$} for (partitioned) probe $H$-free graphs. 
For some graphs $H$, such as $H=P_4$~\cite{chang2005ontherecognition}, probe $H$-free graphs can be recognized in polynomial time.
However, for most graphs~$H$, including $H=P_5$, the recognition of probe $H$-free graphs and the distinction between probes and non-probes are open problems.
Hence, we usually require the sets of probes $P$ and non-probes $N$ to be part of the input, i.e., we must consider partitioned probe $H$-free graphs. 
Note that we can colour a probe $H$-free graph $G$ with one extra colour (assigned to each vertex in $N$) than the number of colours used for $G[P]$. The challenge is to determine if we need that extra colour. 

\subsection*{Related Work}
So far, most previous work on probe graphs focused on characterising and recognising classes of probe graphs~\cite{golumbic2004chordalprobe,chang2005ontherecognition,chandler2009onprobepermutationgraphs,bayer2009probethresholdprobetriviallyperfect,golumbic2011chainprobegraphs,berry2007recognizingchordalprobegraphs,le2015recognizingprobeblockgraphs}. However, recently, the first systematic studies of optimisation problems on partitioned probe $H$-free graphs were undertaken. Brettell et al.~\cite{BrettellOPPRL25} considered \textsc{Vertex Cover} 
 on partitioned probe $H$-free graphs 
 and Dabrowski et al.~\cite{DEJPP25} did the same for {\sc Matching Cut} and some of~its variants. A takeaway from~\cite{BrettellOPPRL25,DEJPP25} is that determining the complexity of {\sc Vertex Cover} for (partitioned) probe $P_5$-free~graphs seems challenging, whereas {\sc Matching Cut} is \NP-complete on partitioned probe $P_5$-free graphs, even though they  are both 
 polynomial-time solvable on $P_5$-free graphs~\cite{LVV14,Fe23}.

Helpful for algorithmic studies as~\cite{BrettellOPPRL25,DEJPP25} is that probe graphs inherit some properties from the graph class they are based on. 
This is also true when studying \Col{}. For example, 
Golumbic and Lipshteyn~\cite{golumbic2004chordalprobe} proved that 
probe chordal graphs are perfect. Hence, we obtain that \Col{} is polynomial-time solvable for probe chordal graphs, as {\sc Colouring} is so for perfect graphs~\cite{groetschel1981ellipsoidmethodcombinatorialoptimization,groetschel1984ellipsoidmethodcombinatorialoptimization}.
In 2012, Chandler et al.~\cite{chandler2012probegraphclasses} conjectured that this holds even for {\sc Colouring} on partitioned probe perfect graphs.
Moreover, the following is known:

\begin{proposition}[\hspace{1sp}\cite{chang2005ontherecognition,BrettellOPPRL25}]\label{prop:probe-cw-mimw}
Let $\mathcal{G}$ be a class of graphs and let $w$ be a fixed integer.\\[-19pt]
\begin{enumerate}[(i)]
\item If $\mathcal{G}$ has clique-width at most $w$, then  $\mathcal{G}_p$ has clique-width at most $2w$.
\item If  $\mathcal{G}$ has mim-width at most $w$, then $\mathcal{G}_p$ has mim-width at most $2w$.
\end{enumerate}
\end{proposition}

\noindent
Hence, as $\mathcal{G} \subseteq \mathcal{G}_p$ holds for every graph class $\mathcal{G}$, we obtain that  a graph class $\mathcal{G}$ has bounded mim-width (clique-width) if and only if $\mathcal{G}_p$ has bounded mim-width (clique-width). However, if the yes-instances for some problem $\Pi$ in ${\cal G}$ have bounded width, this may no longer hold for the yes-instances for $\Pi$ in ${\cal G}_p$. If we can still solve $\Pi$ on ${\cal G}_p$, this means there might be a deeper reason for the polynomial-time behaviour of $\Pi$ on ${\cal G}$. We will also research this.

\subsection*{Our Results}
As mentioned, we focus on \Col{}\ and $k$-{\sc Colouring} for probe $H$-free graphs.
Our first result is a full dichotomy of \Col{} on partitioned probe $H$-free graphs (for two graphs $G_1$ and $G_2$, we write $G_1\subseteq_i G_2$ if $G_1$ is an induced subgraph of $G_2$). 

\begin{theorem}[restate=ThmColouringProbe]\label{thm:col-probe}
For a graph~$H$, \ProbeCol{} is polynomial-time solvable for probe $H$-free graphs if $H\subseteq_i P_4$, and else it is \NP-complete even for partitioned probe $H$-free graphs. 
\end{theorem}

\noindent
By Theorem~\ref{thm:col-probe}, \Col{} is \NP-complete for partitioned probe $3P_1$-free graphs, while \Col{} is even polynomial-time solvable on $(P_3+P_1)$-free graphs~\cite{kral2001complexity}.
It is known that the class of $H$-free graphs has bounded mim-width~\cite{brettel2022boundingmimwidth} if and and only if it has bounded clique-width (see e.g.~\cite{dabrowski2016cliquewidth}) if and only if $H$ is an induced subgraph of $P_4$.
Hence, Theorem~\ref{thm:col-probe} also implies, together with Proposition~\ref{prop:probe-cw-mimw},  that \ProbeCol{} on (not necessarily partitioned) probe $H$-free graphs is solvable in polynomial time exactly when the mim-width or clique-width is bounded. In fact, we apply Proposition~\ref{prop:probe-cw-mimw} to show the first part of Theorem~\ref{thm:col-probe}, while for the second part we modify a known hardness reduction for \Col{}~\cite{blanche2019hereditary}; see Section~\ref{app:prop-col-probe}.

Our second result is a full  dichotomy for \kCol{$3$} on partitioned probe 
$P_t$-free~graphs:

\begin{theorem}[restate=ThmcolouringProbePFive]\label{thm:3col-probe-p5}
For an integer $t\geq 1$, \ProbekCol{$3$} on partitioned probe $P_t$-free graphs is polynomial-time solvable if $t\leq 5$ and \NP-complete if $t\geq 6$.
\end{theorem}

\noindent
In Section~\ref{sec:probe-p5} we prove the polynomial part of Theorem~\ref{thm:3col-probe-p5} by giving our {\bf main result}: {a polynomial-time algorithm for \kCol{$3$} for partitioned probe $P_5$-free graphs}. The class of $P_5$-free graphs has been extensively studied for many well-known graph problems, yielding polynomial-time algorithms not only for $k$-{\sc Colouring} for any $k\geq 1$~\cite{HKLSS10}, but also  {\sc Vertex Cover}~\cite{LVV14}, {\sc Feedback Vertex Set}~\cite{ACPRS24}, {\sc Independent Feedback Vertex Set}~\cite{BDFJP19} and very recently, {\sc Odd Cycle Transversal}~\cite{ALLRSS}, or more generally, {\sc Maximum Partial List $H$-Coloring} for every fixed graph~$H$~\cite{LRSSZ}. Our polynomial-time result for $3$-{\sc Colouring} on partitioned probe $P_5$-free graphs is the {\it first} result that generalizes a known polynomial-time result for a (classical) graph problem on $P_5$-free graphs to partitioned probe $P_5$-free graphs.

In Section~\ref{sec:probe-p6} we prove the second part of Theorem~\ref{thm:3col-probe-p5}. In fact, we show that \ProbekCol{$3$} is \NP-complete even on partitioned probe $(P_6,2P_3,3P_2)$-free graphs.
In contrast, \kCol{$3$} is polynomial-time solvable even on $P_7$-free graphs~\cite{BCMSSZ18} and $sP_2$-free graphs for all $s\geq 1$~\cite{DLRR12}. 

In Section~\ref{s-con}, we point out {\it many} natural directions for future research. In particular, we determine all (disconnected) graphs $H$ for which \kCol{$3$} on probe partitioned $H$-free graphs is still open and solve one such open case, namely when $H=P_3+sP_1$. Moreover, we consider \kCol{$k$} for $k\geq 4$ and solve one open case, namely when $H=P_2+sP_1$ for $s\geq 1$, by proving that for every $s\geq 0$, all probe $(P_2+sP_1)$-free graphs are $(s+1)P_2$-free.

\vspace*{-3mm}
\subparagraph{Proof Ideas behind Our Main Result}
As evidenced by the $C_5$, there are $P_5$-free graphs that are not perfect, and there exist three different proofs for showing that $3$-{\sc Colouring} is polynomial-time solvable on $P_5$-free graphs. As discussed below, these proofs are not applicable to probe $P_5$-free graphs, even though two of them provide good starting points.

First, in the proof of Brettell et al.~\cite{brettel2022listkcolouring}, it was shown that $k$-colourable $P_5$-free graphs have {\it bounded mim-width} for all $k\geq 1$, directly implying that $k$-{\sc Colouring} becomes polynomial-time solvable~\cite{BTV13}. However, bipartite graphs are even $2$-colourable and readily seen to be even probe $2P_2$-free (as shown below explicitly for paths) while already chordal bipartite graphs have not only unbounded mim-width~\cite{BCM15} but even {\it unbounded sim-width}~\cite{BBMP24}. This also shows that probe $3$-colourable $P_5$-free graphs, which have bounded mim-width due to Proposition~\ref{prop:probe-cw-mimw}, are only a small subclass of  the class of $3$-colourable probe $P_5$-free graphs.

Second, in the proof of Randerath, Schiermeyer and Tewes~\cite{RST02}, a connected $3$-colourable $P_5$-free graph is shown to have a dominating set~$D$ of constant size. By precolouring~$D$ in every possible way, polynomially many instances of $2$-{\sc List Colouring} (i.e., where all lists of admissible colours have size~$2$) are obtained. As {\sc $2$-List Colouring} can be formulated as $2$-SAT, it is polynomial-time solvable~\cite{Edwards86}. This approach fails for probe $P_5$-free graphs. To see this, let $G$ be a path $u_1u_2\ldots u_{2n}$. Let $P=\{u_1,u_3,\ldots,u_{2n-1}\}$ and $N=
\{u_2,u_4,\ldots,u_{2n}\}$. Note that $P$ and $N$ are independent. Hence, making $N$ a clique yields a split graph, which is $2P_2$-free, so $G$ is even probe $2P_2$-free. However, for large $n$, $G$ has no small dominating set.

Finally, in the proof of Woeginger and Sgall~\cite{WS01}, it is first shown that every $(C_3,P_5)$-free graph $G$ is $3$-colourable. Next, the case where $G$ contains at least one triangle~$C$ (which may not be dominating) is considered. After colouring $C$, the following rule is applied exhaustively: whenever a vertex~$v$ has two neighbours not coloured alike, give $v$ the third colour. Afterwards, it is again shown that an instance of $2$-{\sc List Colouring} is obtained.
Also this approach does not work for probe $P_5$-free graphs due to their more intricate structure. By extending the arguments as in~\cite{WS01}, we can show  that even all $C_3$-free probe $P_5$-free graphs, which include all (probe) $(C_3,P_5)$-free graphs, are
$3$-colourable (see Section~\ref{sec:probe-p5}).  However, $C_3$-free probe $P_5$-free graphs form only a small subclass of $3$-colourable probe $P_5$-free graphs.

As we explain below, our proof for partitioned probe $P_5$-free graphs turns out to be substantially more involved than the second and third proofs for $P_5$-free graphs and does not rely on the boundedness of some width parameter.

First, if all connected components of the subgraph induced by the probes are bipartite, we can colour the probes with colours~$1$ and~$2$ and use colour~$3$ for the non-probes (as the non-probes form an independent set). Next, we show that at most one connected component~$K$ of the subgraph induced by the set of probes can be non-bipartite. Then, by $P_5$-freeness, we find that $K$ has a short odd cycle $C$. The fact that in general, $C$ is not a dominating cycle of the graph causes several complications. To overcome these complications,  we branch on the colours of $C$ and then just like~\cite{WS01}, we try to extend the colouring as much as possible using propagation rules, with the aim to reduce eventually to polynomially many instances of \textsc{$2$-List Colouring}. We show that  the latter is possible via (i) a new decomposition of  probe $P_5$-free graphs, which carefully takes into account the fact that we do not know the missing edges between the non-probes that make the graph $P_5$-free (as illustrated in Figure~\ref{fig:partitionKandM}) and (ii) an adaptation of the standard 2-SAT formula for {\sc $2$-List Colouring}.

\section{Preliminaries}\label{sec:pre}

Let $G$ be a graph, and $k$ be a positive integer.
The \emph{order} of $G$ is its number of vertices, and the \emph{size} of $G$ is its number of edges.
For a vertex $v$ of $G$, we denote its \emph{(open) neighbourhood} by $N_G(v)$, and its \emph{closed neighbourhood} by $N_G[v] = N_G(v) \cup \{v\}$.
For a set $S$ of vertices of $G$, let $N_G[S] = \bigcup_{v \in S} N_G[v]$, and $N_G(S) = N_G[S] \setminus S$.
A vertex $v \notin S$ is \emph{complete} to a set of vertices $S$ if $v$ is adjacent to every vertex of $S$, and $v$ is \emph{anticomplete} to $S$ if $v$ is not adjacent to any vertex of $S$.
Let $S'$ be another set of vertices of $G$ that is disjoint to $S$. If every vertex of $S$ is complete (anticomplete) to $S'$, then $S$ is {\it complete} ({\it anticomplete}) to $S'$.
We write $G[S]$ for the subgraph of $G$ induced by $S$.
For two vertex-disjoint graphs $G_1$ and $G_2$, we let $G_1+G_2$ denote their disjoint union, which is the graph $(V(G_1)\cup V(G_2), E(G_1)\cup E(G_2))$. 
For a graph $G$ and integer $s \geq 1$, $sG$ denotes the disjoint union of $s$ copies of $G$.

For a graph~$H$, we say that $G$ is \emph{$H$-free} if there is no set of vertices $S$ such that $G[S]$ is isomorphic to $H$.
We say that $G$ is \emph{probe $H$-free} if there is a partition of the vertices of $G$ into a set of probes $P$ and a set of non-probes $N$, such that $N$ is independent in $G$, and there is a set of edges $F \subseteq {N \choose 2}$ such that $G+F$ is $H$-free.
Note that $G[P]$ is $H$-free if $G$ is probe $H$-free.
A \emph{partitioned probe $H$-free graph} is a triple $(G,P,N)$, where $G$ is a probe $H$-free graph with $P$ as the probes and $N$ as the non-probes, that is, the sets of probes and non-probes are given.
For a set $\{H_1,\ldots,H_r\}$ of graphs, a graph $G$ is \emph{$(H_1,\ldots,H_r)$-free} if $G$ is $H_i$-free for every $i\in \{1,\ldots,r\}$.
A graph $G$ is {\it probe $(H_1,H_2,\ldots)$-free} if there is an independent set $N$ of non-probes in $G$ and a set of edges $F \subseteq {N \choose 2}$ such that $G+F$ is $(H_1,H_2,\ldots)$-free.
In a partitioned probe $(H_1,H_2,\ldots)$-free graph $(G,P,N)$, the graph $G$ is probe $(H_1,H_2,\ldots)$-free with set of probes $P$ and set of non-probes $N$.

We define $[k] = \{1, \dots, k\}$.
A \emph{partial $k$-colouring} of $G$ is a function $\psi: V(G) \rightarrow [k] \cup \{\perp\}$ such that, if $uv \in E(G)$ with $\psi(u),\psi(v) \in [k]$, then $\psi(u) \neq \psi(v)$.
If $v$ is a vertex of $G$ with $\psi(v) \in [k]$, then $v$ is \emph{coloured} (under $\psi$).
Let $\psi'$ be another partial $k$-colouring of $G$.
Then $\psi'$ is an \emph{extension} of $\psi$ if $\psi(v) \in [k]$ implies that $\psi'(v) = \psi(v)$.
A \emph{$k$-colouring} of $G$ is a partial $k$-colouring under which every vertex of $G$ is coloured.
For $S \subseteq V(G)$, we define $\psi(S) = \{\psi(v) : v \in S\}$.

Algorithm~\ref{alg:prop} is a simple colour propagation algorithm that is essential to the proof of Theorem~\ref{thm:3col-probe-p5}. The following properties of Algorithm~\ref{alg:prop} are easy and their proofs are omitted:

\begin{algorithm}
\SetKw{GoTo}{goto}
\DontPrintSemicolon
\KwIn{A graph $G$, and a partial $k$-colouring $\psi$.}
\KwOut{An extension of $\psi$, or an error.}
\medskip
\tcp{Propagation Rule}
\While{there is an uncoloured vertex $v \in V(G)$ and $i \in [k]$ such that $v$ has a neighbour of every colour except colour $i$, that is, $[k] \setminus \{i\} \subseteq \psi(N_G(v)) \subseteq ([k] \setminus \{i\}) \cup \{\perp\}$}{
		set $\psi(v) \gets i$\;
}
\smallskip
\ForAll{$v \in V(G)$}{
	\If{$v$ has a neighbour of every colour, that is, $[k] \subseteq \psi(N_G(v))$}{
		\Return{an error}
	}	
}
\smallskip
\Return{$\psi$}
\caption{Simple colour propagation.}\label{alg:prop}
\end{algorithm}

\begin{lemma}\label{lem:prop}
Let $\psi$ be a partial $k$-colouring of a graph $G$.
\begin{enumerate}[(i)]
\item If Algorithm~\ref{alg:prop} on $(G, \psi)$ returns an extension $\psi'$ of $\psi$ and $v \in V(G)$ is coloured under $\psi'$, then $v$ has the same colour under any $k$-colouring of $G$ that is an extension of $\psi$ (if any exist).
\item If Algorithm~\ref{alg:prop} on $(G, \psi)$ returns an error, then there is no $k$-colouring of $G$ that is an extension of $\psi$.
\item Algorithm~\ref{alg:prop} runs in polynomial time.
\end{enumerate}
\end{lemma}

We use the following well-known lemma, which is due to Edwards~\cite{Edwards86}. We provide a proof to adapt it later.

\begin{lemma}\label{lem:ext}
Given a graph $G$ and a partial $k$-colouring $\psi$ of $G$, for every uncoloured vertex $v \in V(G)$, define the set of available colours of $v$ as $L(v) = [k] \setminus \psi(N_G(v))$.
If $|L(v)| \leq 2$ for every uncoloured vertex $v \in V(G)$, then deciding if there is a $k$-colouring that is an extension of $\psi$ is possible in polynomial time.
\end{lemma}
\begin{proof}
Let the SAT formula $\mathcal{F}$ in conjunctive normal form have variables $x_v^i$ for every uncoloured vertex $v \in V(G)$ and every $i \in L(v)$, and clauses
\begin{itemize}
\item $\bigvee_{i \in L(v)} x_v^i$ for every uncoloured vertex $v$ (i.e., if $L(v) = \emptyset$, then $\mathcal{F}$ is not satisfiable) and
\item $\bar{x}_u^i \vee \bar{x}_v^i$ for every $uv \in E(G)$ with uncoloured vertices $u$ and $v$ and $i \in L(u) \cap L(v)$.
\end{itemize}
According to the assumptions $\mathcal{F}$ is a 2-SAT formula.
By construction, there is a $k$-colouring of $G$ that is an extension of $\psi$ if and only if $\mathcal{F}$ is satisfiable.
This completes the proof since deciding the satisfiability of a 2-SAT formula is possible in polynomial time~\cite{aspvall1979lineartime}.
\end{proof}

\section{The Proof of Theorem~\ref{thm:col-probe}}\label{app:prop-col-probe}

In this section we show Theorem~\ref{thm:col-probe}. The proof of our next result is based on an existing construction from~\cite{blanche2019hereditary}.

\begin{proposition}\label{prp:col-probe-p1p3}
\ProbeCol{} is \NP-complete on partitioned probe $3P_1$-free graphs.
\end{proposition}
\begin{proof}
Clearly, \ProbeCol{} is in \NP. 
Our \NP-hardness reduction is the same as the one Blanch\'{e} et al.~\cite[Theorem~6]{blanche2019hereditary} used to prove that \Col{} is \NP-complete for, amongst others, $(P_6,\overline{P_6})$-free graphs (where $\overline{P_6}$ is the complement of $P_6$). We must repeat their gadget below in order to show that it is a probe $3P_1$-free graph.
We use the \ExactThreeCover{} problem, which is well known to be \NP-complete~\cite{garey1979computers}.
\begin{framed}
\noindent
\ExactThreeCover{}\\
\textit{Input:} A finite set $X$ and a collection $\mathcal{S}$ of 3-subsets of $X$.\\
\textit{Question:} Is there a subcollection $\mathcal{S}' \subseteq \mathcal{S}$ such that each element of $X$ occurs in exactly one subset in $\mathcal{S}'$?
\end{framed}

\begin{figure}[b]
\centering
\begin{tikzpicture}
\foreach \i in {1,...,6}{
	\node[dot,label=above:$\i$] (x\i) at (\i,2) {};
};
\node[dot] (z1) at (3.5,0) {};

\foreach [count=\i] \s in {123,456,235}{
	\node[dot] (y\i) at (2*\i-0.5,1) {};
	\draw (z1) to (y\i);
};

\draw[rounded corners] (0.5,1.666) rectangle (6.5,2.666);
\node at (7,2) {$X$};
\draw[rounded corners] (0.5,0.666) rectangle (6.5,1.333);
\node at (7,1) {$Y$};
\draw[rounded corners] (0.5,-0.333) rectangle (6.5,0.333);
\node at (7,0) {$Z$};

\draw (y1) to (x1);
\draw (y1) to (x2);
\draw (y1) to (x3);
\draw (y2) to (x2);
\draw (y2) to (x3);
\draw (y2) to (x5);
\draw (y3) to (x4);
\draw (y3) to (x5);
\draw (y3) to (x6);
\end{tikzpicture}
\caption{
The graph $G$ constructed form the instance $([6], \{\{1,2,3\},\{2,3,5\},\{4,5,6\}\})$ of \ExactThreeCover{}.
We omitted drawing the edges of the clique $X$.
}
\label{fig:ETC}
\end{figure}
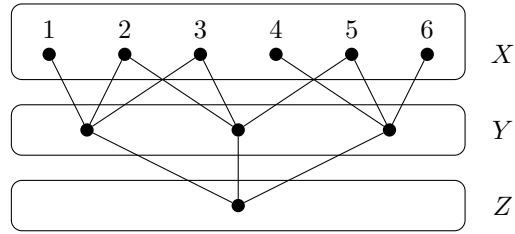

To prove the \NP-hardness of \ProbeCol{}, we reduce from the \NP-complete problem \ExactThreeCover{}~\cite{garey1979computers}.
To this end, let $(X, \mathcal{S})$ be an instance of \ExactThreeCover{} and $s = |\mathcal{S}|$.
We may assume that $|X| = 3k$ for a positive integer $k$ and $s \geq k$; otherwise, $(X, \mathcal{S})$ is a no-instance and we map it to some trivial no-instance of \ProbeCol{}.
Let $G$ be the graph defined as follows.
The vertex set of $G$ is the disjoint union of the sets $X$, $Y$, and $Z$, where $Y = \{y_S: S \in \mathcal{S}\}$ and $|Z| = s - k$.
The set $X$ induces a clique in $G$, while $Y$ and $Z$ are both independent in $G$.
The set $X$ is anticomplete to $Z$, while the set $Y$ is complete to $Z$.
Between $X$ and $Y$ there are exactly the edges $xy_S$ for $x \in X$ and $y_S \in Y$ with $x \in S$.
This completes the description of $G$; see Figure~\ref{fig:ETC}. 
Clearly, $G$ is constructable in polynomial time.

We claim that $(X, \mathcal{S})$ is a yes-instance of \ExactThreeCover{} if and only if the vertex set of $G$ is the union of $s$ pairwise disjoint cliques.
If $\mathcal{S'} \subseteq \mathcal{S}$ is such that each element of $X$ is contained in exactly one subset of $\mathcal{S'}$, then a covering of $G$ with $s$ pairwise disjoint cliques is given by $S \cup \{y_S\}$ for each $S \in \mathcal{S'}$ and the edges of a perfect matching between $\{y_S: S \in \mathcal{S} \setminus \mathcal{S'}\}$ and $Z$, which exists.
For the other direction, let $V_1, \dots, V_s$ be pairwise disjoint cliques of $G$ such that $V(G) = \bigcup_{i \in [s]} V_i$. 
Let $I = \{i \in [s]: V_i \cap Z \neq \emptyset\}$.
Since $Y$ is independent in $G$, each $V_i$ contains exactly one vertex of $Y$.
Similarly, as $Z$ is independent in $G$, we have $|I| = s-k$.
Since $X$ and $Z$ are anticomplete, we have $X \subseteq \bigcup_{i \in [s] \setminus I} V_i$.
Now, since every vertex $y_S \in Y$ has exactly 3 neighbours in $X$, the sets $V_i$ for $i \in [s] \setminus I$ have cardinality at most~$4$.
Since $|[s] \setminus I| = k$, we get that they have cardinality exactly~$4$.
Since the $V_i$ are pairwise disjoint,
\[
\mathcal{S'} = \left\{S \in \mathcal{S}: y_S \in \bigcup_{i \in [s] \setminus I} V_i\right\}
\]
witnesses that $(X, \mathcal{S})$ is a yes-instance of \ExactThreeCover.

At this point, consider the complement $\overline{G}$ of $G$.
Note that the complement is computable in polynomial time.
Observe that $X$ is independent in $\overline{G}$.
The graph $\overline{G}$ is probe $3P_1$-free, since $\overline{G} + {X \choose 2}$, the graph obtained from $\overline{G}$ by turning $X$ into a clique, is $3P_1$-free.
To see this, observe that $X \cup Z$ and $Y$ are cliques in $\overline{G} + {X \choose 2}$.
The fact that the vertex set of $G$ is the union of $s$ pairwise disjoint cliques if and only if $(\overline{G}, Y \cup Z, X, s)$ is a yes-instance completes the proof.
\end{proof}

\noindent 
We combine Proposition~\ref{prp:col-probe-p1p3} with the \NP-completeness part of the dichotomy of  Kr{\'{a}}l et al.~\cite{kral2001complexity}, the fact that $P_4$-free graphs have clique-width~$2$~\cite{courcelle2000upperbounds} and Proposition~\ref{prop:probe-cw-mimw} to obtain Theorem~\ref{thm:col-probe}:

\ThmColouringProbe*
\begin{proof}
We first recall the definition of cliquewidth and $k$-expressions. A \emph{$k$-expression} combines any number of the following operations on a labelled graph with labels in $[k]$:
\begin{itemize}
\item create a new graph with a single vertex with label $1$;
\item given $i,j \in [k]$, $i \not= j$, relabel all vertices with label $i$ to label $j$;
\item given $i,j \in [k]$, $i \not= j$, add all edges between vertices with label $i$ and label $j$;
\item take the disjoint union of two labelled graphs with labels in $[k]$.
\end{itemize}
The smallest integer~$k$ for which a graph~$G$ has a $k$-expression is called the  \emph{clique-width} of $G$~\cite{CourcelleER93}.

We now continue as follows. By the results of Kr{\'{a}}l et al.~\cite{kral2001complexity}, \Col{} on $H$-free graphs is \NP-complete unless $H$ is an induced subgraph of $P_4$ or $P_3+P_1$, and thus \ProbeCol{} on partitioned probe $H$-free graphs is \NP-complete unless $H$ is an induced subgraph of $P_4$ or $P_3+P_1$.

If $H$ is an induced subgraph of $P_4$, we obtain a polynomial-time algorithm for probe $H$-free graphs as well. Since $P_4$-free graphs are the graphs with clique-width at most~$2$~\cite{courcelle2000upperbounds}, probe $P_4$-free graphs have clique-width at most~$4$ by Proposition~\ref{prop:probe-cw-mimw}~(i). We can find a $15$-expression of probe $P_4$-free graphs in polynomial time~\cite{HlinenyO08} and solve \Col{} in polynomial time~\cite{EspelageGW01,KoblerR03}.

Note that the only graphs that are an induced subgraph of $P_3+P_1$, but not an induced subgraph of $P_4$, are $3P_1$ and $P_3+P_1$. Hence, Proposition~\ref{prp:col-probe-p1p3} shows that \ProbeCol{} is \NP-complete for those cases.
\end{proof}

\section{The Proof of the Polynomial Part of Theorem~\ref{thm:3col-probe-p5}}\label{sec:probe-p5}

In this section we prove the main result of our paper, namely that $3$-{\sc Colouring} is polynomial-time solvable for partitioned probe $P_5$-free graphs.

We first show the following independent result (Proposition~\ref{p-easy}) in more or less the same way as done in~\cite{WS01} for showing that $(C_3,P_5)$-free graphs are $3$-colourable. The main difference is that we must take into account the probes and non-probes. As such, the proof of this result serves as warm-up exercise illustrating some of the arguments we will use in a more involved way in the proof of our main result. Note that probe $(C_3,P_5)$-free graphs form a subclass of $C_3$-free probe $P_5$-free graphs. This containment is proper, as
\begin{itemize}
\item [(i)] probe $(C_3,P_5)$-free graphs have bounded mim-width, due to $(C_3,P_5)$-free graphs having bounded mim-width~\cite{brettel2022listkcolouring} and Proposition~\ref{prop:probe-cw-mimw}, and
\item [(ii)] chordal bipartite graphs, which are $C_3$-free probe $P_5$-free, have unbounded mim-width~\cite{BCM15}.
\end{itemize}

\begin{proposition}\label{p-easy}
All $C_3$-free probe $P_5$-free graphs, and thus all probe $(C_3,P_5)$-free graphs, are $3$-colourable.
\end{proposition}

\begin{proof}
Let $G=(V,E)$ be a $C_3$-free probe $P_5$-free graph. Let $P$ be the set of probes and $N=V\setminus P$ be the set of non-probes, so $N$ is an independent set. By definition, there exists a set $F$ of edges with both end-vertices in $N$ such that $G+F$ is $P_5$-free. We may assume without loss of generality that $G$ is connected, as otherwise we consider every connected component of $G$ separately.

If $G[P]$ is bipartite, then we colour $G[P]$ with colours~$1$ and~$2$, and as $N$ is independent, we can use colour~$3$ for the vertices of~$N$.
Now suppose that $G[P]$ is not bipartite. This means that $G[P]$ has an odd cycle $C$. As $G$ is $C_3$-free and probe $P_5$-free, $G[P]$ must be $(C_3,P_5)$-free. Hence, $C$ has length~$5$. Let $V(C)=\{v_1,\ldots,v_5\}$ in that order. 
We will now use, with a little bit of extra care, the same arguments as in~\cite{WS01}.

We first claim that every vertex not on $C$ has a neighbour on $C$. Else, as $G$ is connected, there exists a vertex~$u\in V\setminus V(C)$ that is anti-complete to $C$ and that has a neighbour $v$ that is adjacent to at least one vertex on $C$, say $vv_1\in E$. If $uvv_1v_2v_3$ is an induced $P_5$ in $G$, then $uvv_1v_2v_3$ is also an induced $P_5$ in $G+F$. The reason is that $F$ contains no edge that is incident with a vertex from $\{v_1,v_2,v_3\}$, because $v_1,v_2,v_3$ all belong to $P$.
As $G$ is $C_3$-free and $G+F$ is $P_5$-free, this means that $v$ must be adjacent to $v_3$ in $G$. By applying the same arguments on the path $uvv_3v_4v_5$, we find that $v$ must be adjacent to $v_5$ in $G$ as well. However, now $v,v_1,v_5$ form a triangle in~$G$, contradicting the $C_3$-freeness of $G$.

We now claim that every vertex not on $C$ has exactly two neighbours $v_i$ and $v_{i+2}$ on $C$ for some $i\in \{1,\ldots, 5\}$, where we write $v_6:=v_1$ and $v_7:=v_2$. Let $v\in V\setminus V(C)$. By the above, $v$ has a neighbour on $C$, say $v$ is adjacent to $v_1$. Recall that $F$ does not contain any edges incident to vertices of $C$, as $V(C)\subseteq P$. Hence, if $vv_1v_2v_3v_4$ is an induced $P_5$ in $G$, then $vv_1v_2v_3v_4$ is also an induced $P_5$ in $G+F$. As $G$ is $C_3$-free and $G+F$ is $P_5$-free, this means that $v$ is either adjacent to $v_3$ (take $i=1$) or to $v_4$ (take $i=4$).
Due to the above, we can decompose $V$ as $$V=V(C)\cup V_{1,3} \cup V_{2,4} \cup V_{3,5} \cup V_{4,1} \cup V_{5,2},$$ where for $i\in \{1,\ldots,5\}$, the set $V_{i,i+2}$ consist of all vertices of $V\setminus C$ whose neighbours on $C$ are exactly $v_i$ and $v_{i+2}$.
We give $v_1,v_2,v_3,v_4,v_5$ colours $1,2,1,2,3$, respectively. Moreover, we colour all the vertices of $V_{1,3}$, $V_{2,4}$, $V_{3,5}$, $V_{4,1}$ and $V_{5,2}$ with colours $2,1,2,3,1$, respectively. As $G$ is $C_3$-free, the five sets $V_{i,i+1}$ are all independent. Therefore, the only potential conflicts could be between vertices from $V_{1,3}$ and $V_{3,5}$, which are all coloured~$2$, or between vertices from $V_{2,4}$ and $V_{5,2}$, which are all coloured~$1$. However, the former vertices are all incident to $v_3$, and thus form an independent set in $G$, and similarly, the latter vertices are all adjacent to $v_2$, and thus als form an independent  set in $G$. We conclude that we have indeed constructed a $3$-colouring of $G$, completing the proof.
\end{proof}

\noindent
We are now ready to show the main result of our paper, which we prove in the way outlined at the end of Section~\ref{sec:intro}.

\begin{theorem}\label{t-poly}
\kCol{$3$} is polynomial-time solvable for partitioned probe $P_5$-free graphs.
\end{theorem}

\begin{proof}
Let $(G,P,N)$ be a partitioned probe $P_5$-free graph.
We may assume that $G$ is connected; otherwise, we execute the given algorithm for every connected component of $G$.
Let $F \subseteq {N \choose 2}$ be such that $G + F$ is $P_5$-free.
We define $F$ only for verifying correctness; the polynomial-time algorithm does not use $F$.
If $G$ is $P_5$-free, then it is possible in polynomial time to determine whether $G$ is 3-colourable~\cite{HKLSS10}.
Therefore, we may assume that $G$ is not $P_5$-free and, in particular, $|N| \geq 2$ and $|F| \geq 1$.
We may also assume that $G$ does not contain a clique of order at least~$4$;
otherwise, $G$ is not $3$-colourable.
Let $K_1, \dots, K_t$ be the connected components of $G[P]$  that contain at least one edge.
We may assume at least one such connected component exists; else $G$ is bipartite with partite sets $P$ and $N$, and thus clearly 3-colourable in polynomial time.

\subparagraph*{Getting initial structure}
We begin by proving two claims that describe the structure of edges between $K_1, \dots, K_t$ and $N$.

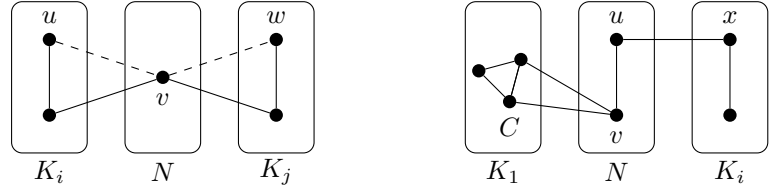
\begin{figure}
\centering
\begin{tikzpicture}
\begin{scope}[shift={(-3,0)}]
\draw[rounded corners] (-2,-1) rectangle (-1,1);
\node at (-1.5,-1.25) {$K_i$};
\draw[rounded corners] (-0.5,-1) rectangle (0.5,1);
\node at (0,-1.25) {$N$};
\draw[rounded corners] (1,-1) rectangle (2,1);
\node at (1.5,-1.25) {$K_j$};

\node[dot,label=above:$u$] (u) at (-1.5,0.5) {};
\node[dot] (nu) at (-1.5,-0.5) {};
\node[dot,label=below:$v$] (v) at (0,0) {};
\node[dot,label=above:$w$] (w) at (1.5,0.5) {};
\node[dot] (nw) at (1.5,-0.5) {};

\draw (u) to (nu) to (v) to (nw) to (w);
\draw[dashed] (u) to (v) to (w);
\end{scope}
\begin{scope}[shift={(3,0)}]
\draw[rounded corners] (-2,-1) rectangle (-1,1);
\node at (-1.5,-1.25) {$K_1$};
\draw[rounded corners] (-0.5,-1) rectangle (0.5,1);
\node at (0,-1.25) {$N$};
\draw[rounded corners] (1,-1) rectangle (2,1);
\node at (1.5,-1.25) {$K_i$};

\node[dot,label=above:$u$] (u) at (0,0.5) {};
\node[dot,label=below:$v$] (v) at (0,-0.5) {};
\foreach \i in {1, ..., 3} {
	\ifthenelse{\i=2}{
		\node[dot,shift={(-1.5,0)},label=below:$C$] (t\i) at (\i*120+45:0.333) {};
	}{
		\node[dot,shift={(-1.5,0)}] (t\i) at (\i*120+45:0.333) {};
	}
};
\node[dot,label=above:$x$] (x) at (1.5,0.5) {};
\node[dot] (y) at (1.5,-0.5) {};
\draw (t1) to (t2) to (t3) to (t1);
\draw (v) to (t3) to (t2) to (v) to (u) to (x) to (y);
\end{scope}
\end{tikzpicture}
\caption{Left: Proof of Claim~\ref{claim:anti-complete}. The dashed lines indicate non-existing edges. Right: Proof of Claim~\ref{claim:complete-to-Ki}. Note that $uv \in F$.}\label{fig:claim1and2}
\end{figure}

\begin{claim}\label{claim:anti-complete}
Every vertex of $N$ that is neither complete nor anticomplete to $K_i$ for some $i \in [t]$ is complete or anticomplete to $K_j$ for every $j \in [t]$ with $j \neq i$.
\end{claim}
\begin{proof}
Let $v \in N$ be neither complete nor anticomplete to $K_i$.
Suppose that $v$ has a neighbour in $K_j$, where $j \neq i$.
It suffices to prove that $v$ is complete to $K_j$.
Assume, for a contradiction, that $w \in V(K_j)$ is not adjacent to $v$.
By assumption, there exists $u \in V(K_i)$ that is not adjacent to $v$.
A shortest $u$-$v$-path with internal vertices in $K_i$ followed by a shortest $v$-$w$-path with internal vertices in $K_j$ is induced in $G+F$ and has length at least 4; see Figure~\ref{fig:claim1and2}.
Since such a path  exists, there is an induced $P_5$ in $G+F$, a contradiction.
\end{proof}

If $K_1, \dots, K_t$ are all bipartite, then $G$ is clearly 3-colourable, since $N$ is independent in $G$.
Therefore, we may assume that $t \geq1$ and $K_1$ is not bipartite.
This implies that $K_1$ contains an induced odd cycle and, since $K_1$ is $P_5$-free because of $V(K_1) \subseteq P$, such a cycle has length~$3$ or length~$5$.
We now pick an induced odd cycle $C$ in $K_1$ as follows. If $K_1$ contains an induced $C_5$, then let $C$ be any such $C_5$. If $K_1$ does not contain an induced $C_5$, but contains an induced $C_3$ that dominates $K_1$, then let $C$ be any such $C_3$. Otherwise, we pick $C$ to be an arbitrary $C_3$.
Note that computing $C$ is possible in polynomial time.

If a single vertex of $V(G) \setminus C$ dominates $C$, then $G$ is clearly not 3-colourable.
Hence, we may assume from here that this is not the case.
This fact (that we often use implicitly) has important implications. In particular, no vertex of $N$ dominates $K_1$. But also:

\begin{claim}\label{claim:complete-to-Ki}
Let $u \in N$ be a vertex with no neighbour in $K_1$.
If $u$ has a neighbour in $K_i$ with $i \geq 2$, then a vertex of $N$ with a neighbour in $K_1$ is complete to $K_i$.
\end{claim}
\begin{proof}
Consider a shortest $u$-$C$-path $Q$ in $G+F$. As $u$ has no neighbour in $K_1$, $Q$ has length at least~$2$. Let $w$ be the vertex of $C$ where $Q$ ends and let $v$ be the vertex on $Q$ preceding $w$. Using the observation preceding the claim, $v$ is not complete to $C$. We may thus assume that $Q$ was chosen such that there exists a vertex $z \in N_C(w) \setminus N_{G+F}(v)$. If $v \in K_1$, then as $u$ does not neighbour $K_1$, the path $Qz$ has length at least~$4$, a contradiction to the fact that $G+F$ is $P_5$-free. Hence, $v \in N \setminus\{u\}$ and $v$ is neither complete nor anticomplete to $K_1$. If $v$ is not a neighbour of $u$ in $G+F$, then $Qz$ is an induced path in $G+F$ of length at least~$4$, a contradiction. Let $x$ be a neighbour of $u$ in $K_i$. If $x$ is not a neighbour of $v$ in $G+F$, then the path $xuvwz$ is an induced $P_5$ in $G+F$, a contradiction; see Figure~\ref{fig:claim1and2} right. Hence, $v$ has a neighbour in $K_i$, and the claim follows from Claim~\ref{claim:anti-complete}.
\end{proof}

\begin{claim} \label{claim:isbipartite}
The connected components $K_2, \dots, K_t$ are all bipartite or $G$ is not 3-colourable.
\end{claim}
\begin{proof}
Assume (without loss of generality) $K_2$ is not bipartite and $G$ is $3$-colourable. From our earlier observation, if some vertex is complete to $K_1$ or to $K_2$, then $G$ is not $3$-colourable, a contradiction. As $G$ is connected, $K_2$ has a neighbour $u \in N$. Hence, $u$ is neither complete or anticomplete to $K_2$. As $u$ cannot be complete to $K_1$, by Claim~\ref{claim:anti-complete}, $u$ is anticomplete to $K_1$. Then, by Claim~\ref{claim:complete-to-Ki}, there is a vertex in $N$ that is complete to $K_2$, a contradiction.
\end{proof}
We can check in linear time whether $K_2, \dots, K_t$ are all indeed bipartite.

\subparagraph*{Colouring $C$}
Let $K=K_1$ for brevity and $I = P \setminus V(K)$.
Note that $G[I]$ consists only of isolated vertices and bipartite connected components.
We branch on all partial 3-colourings $\psi$ that only colour every vertex of $C$. 
There are constantly many branches, as there are only constantly many such partial 3-colourings. 
We propagate the colours through $K$ by executing Algorithm~\ref{alg:prop} on $(K, \psi)$.
If an error occurred, then there is no 3-colouring of $G$ that is an extension of $\psi$ by Lemma~\ref{lem:prop}~(ii), and we backtrack.
So we may assume that no error occurred, and for simplicity we denote the returned extension of $\psi$ by $\psi$ again.

We explicitly only propagated the colours through $K$. We now partition of $V(K)$. Let

\begin{itemize}
\item $K_c^i$ be the set of vertices of $K$ with colour $i$ for $i \in [3]$,
\item $K_c = \bigcup_{i \in [3]} K_c^i$,
\item $K_u^i$ be the set of uncoloured vertices of $K$ with a neighbour of colour $i$ for $i \in [3]$,
\item $K_u = \bigcup_{i \in [3]} K_u^i$, and
\item $K_r = V(K) \setminus (K_c \cup K_u)$ consist of the remaining vertices of $K$.
\end{itemize}
Note that $G[K_c]$ is connected, because $C$ is connected, and we assign colours to uncoloured vertices only with the Propagation Rule in Algorithm~\ref{alg:prop}.
Also note that the vertices of $K_u^i$ have only neighbours of colour $i \in [3]$ since they are uncoloured.

Our ultimate goal is to apply Lemma~\ref{lem:ext}.
So far we are not in a position to apply it, since there may be vertices (for example in $K_r$) that do not have a coloured neighbour. In the remaining proof, we distinguish two cases, depending on the length of $C$.

\subparagraph*{Case 1: $C$ has length 5}
We show that all vertices already have a coloured neighbour.

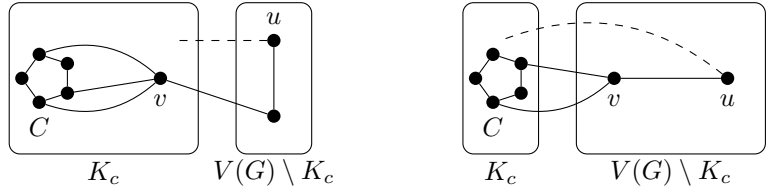
\begin{figure}
\centering
\begin{tikzpicture}
\begin{scope}[shift={(-3,0)}]
\draw[rounded corners] (-2,-1) rectangle (0.5,1);
\node at (-0.75,-1.25) {$K_c$};
\draw[rounded corners] (1,-1) rectangle (2,1);
\node at (1.5,-1.25) {$V(G) \setminus K_c$};

\foreach \i in {1, ..., 5} {
	\ifthenelse{\i=1}{
		\node[dot,shift={(-1.5,0)},label=below:$C$] (t\i) at (\i*72+180:0.333) {};
	}{
		\node[dot,shift={(-1.5,0)}] (t\i) at (\i*72+180:0.333) {};
	}
};
\draw (t1) to (t2) to (t3) to (t4) to (t5) to (t1);

\node[dot,label=below:$v$] (v) at (0,0) {};
\draw[bend right] (t1) to (v);
\draw (t2) to (v);
\draw[bend left] (t4) to (v);
\node[dot,label=above:$u$] (u) at (1.5,0.5) {};
\node[dot] (nu) at (1.5,-0.5) {};
\draw (u) to (nu) to (v);
\draw[dashed] (u) to (0.25,0.5);
\end{scope}
\begin{scope}[shift={(3,0)}]
\draw[rounded corners] (-2,-1) rectangle (-1,1);
\node at (-1.5,-1.25) {$K_c$};
\draw[rounded corners] (-0.5,-1) rectangle (2,1);
\node at (0.75,-1.25) {$V(G) \setminus K_c$};

\foreach \i in {1, ..., 5} {
	\ifthenelse{\i=1}{
		\node[dot,shift={(-1.5,0)},label=below:$C$] (t\i) at (\i*72+180:0.333) {};
	}{
		\node[dot,shift={(-1.5,0)}] (t\i) at (\i*72+180:0.333) {};
	}
};
\draw (t1) to (t2) to (t3) to (t4) to (t5) to (t1);

\node[dot,label=below:$v$] (v) at (0,0) {};
\draw[bend right] (t1) to (v);
\draw (t3) to (v);
\node[dot,label=below:$u$] (u) at (1.5,0) {};
\draw (u) to (v);
\draw[dashed,bend right] (u) to (-1.5,0.5);
\end{scope}
\end{tikzpicture}
\caption{
Proof of Claim~\ref{claim:coloured-nbor}.
Dashed lines indicate non-existing edges.
}\label{fig:coloured-nbor}
\end{figure}

\begin{claim}\label{claim:coloured-nbor}
Every vertex of $V(G) \setminus K_c$ has a neighbour in $K_c$.
\end{claim}
\begin{proof}
Assume, for a contradiction, that $u \in V(G) \setminus K_c$ has no neighbour in $K_c$.
Consider a shortest $u$-$C$-path $Q$ in $G+F$.
Let $v$ be the vertex of $Q$ that has a neighbour in $C$.
Note that $v$ is not complete to $C$; otherwise, we would have concluded that $G$ is not 3-colourable.
If $v$ is in $K_c$ itself, then $Q$ has length at least 3, and there would be an induced $P_5$ in $G+F$ with vertices in $V(Q) \cup V(C)$; see Figure~\ref{fig:coloured-nbor} left.
Hence, $v$ is not in $K_c$. Then $v$ has at most two neighbours in $C$, and $Q$ has length at least 2, and there would be an induced $P_5$ in $G+F$ with vertices in $V(Q) \cup V(C)$, a contradiction; see Figure~\ref{fig:coloured-nbor} right.
\end{proof}

Claim~\ref{claim:coloured-nbor} implies that Lemma~\ref{lem:ext} is applicable in this case.
Therefore, deciding if there is a 3-colouring of $G$ that is an extension of $\psi$ is possible in polynomial time.
If there is no such 3-colouring of $G$, then we backtrack.

\subparagraph*{Case 2: $C$ has length 3}
First, note that for every vertex $v \in K_c$, we have that $v$ has two neighbours with two distinct colours in $[3] \setminus \{\psi(v)\}$, since $C$ is a clique and we assign colours to uncoloured vertices only through the Propagation Rule in Algorithm~\ref{alg:prop}.
We now give a more precise partition of $N$; see Figure~\ref{fig:partitionKandM}.
Let $M = N_G(I)$ and $L = N \setminus M$. Let
\begin{itemize}
\item $M_c$ and $L_c$ be the set of vertices of $M$ and $L$, respectively, that have two neighbours in $K_c$ with two distinct colours,
\item $M_u^i$ and $L_u^i$ be the set of vertices of $M \setminus M_c$ and $L \setminus L_c$, respectively, with a neighbour in $K_c^i$ for $i \in [3]$,
\item $M_u = \bigcup_{i \in [3]} M_u^i$, $L_u = \bigcup_{i \in [3]} L_u^i$,
\item $M_r = M \setminus (M_c \cup M_u)$, and $L_r = L \setminus (L_c \cup L_u)$.
\end{itemize}
Let $J$ be the set of vertices of $I$ with no neighbour in $M_c$.
Note that no vertex of $K_r$, $L_r$, $M_r$, and $J$ has a coloured neighbour. We now show how in the end we can apply Lemma~\ref{lem:ext}.

\begin{figure}
\centering
\begin{tikzpicture}
\draw[rounded corners] (-6.5,2.75) rectangle (-5.5,-3.5);
\node at (-6,-3.75) {$L$};
\draw[rounded corners] (-5.25,2.75) rectangle (-0.75,-3.5);
\node at (-3,-3.75) {$K$};
\draw[rounded corners] (-5,0.5) rectangle (-4,2.5);
\node at (-4.5,0.25) {$K_c$};
\draw[rounded corners] (-3.5,0) rectangle (-2.5,1);
\node at (-3,-0.25) {$K_u^1$};
\draw[rounded corners] (-3.5,-1.5) rectangle (-2.5,-0.5);
\node at (-3,-1.75) {$K_u^2$};
\draw[rounded corners] (-3.5,-3) rectangle (-2.5,-2);
\node at (-3,-3.25) {$K_u^3$};
\draw[rounded corners] (-2,-2) rectangle (-1,0);
\node at (-1.5,-2.25) {$K_r$};
\draw[rounded corners] (-0.5,-3.5) rectangle (0.5,2.75);
\node at (0,-3.75) {$M$};
\draw[rounded corners] (1,-3.5) rectangle (2,2.75);
\node at (1.5,-3.75) {$I$};

\foreach \i in {1, ..., 3} {
	\ifthenelse{\i=3}{
		\node[dot,shift={(-4.5,1.5)},label=above:$C$] (t\i) at (\i*120+45:0.333) {};
	}{
		\node[dot,shift={(-4.5,1.5)}] (t\i) at (\i*120+45:0.333) {};
	}
};
\draw (t1) to (t2) to (t3) to (t1);

\node[dot] (ku1) at (-3,0.5) {};
\draw (t3) to (ku1);

\node[dot] (ku2) at (-3,-1) {};
\draw (t2) to (ku2);

\node[dot] (ku3) at (-3,-2.5) {};
\draw (t1) to (ku3);

\node[dot] (kr) at (-1.5,-1) {};

\node[dot,label=below:{$M_c$}] (mc) at (0,2.25) {};
\draw (t3) to (mc) to (t2);
\draw (-0.5,1.25) to (0.5,1.25);
\node[dot,label=below:{$M_u^1$}] (mu1) at (0,0.75) {};
\draw (mu1) to (t3);
\draw (-0.5,-0.25) to (0.5,-0.25);
\node at (0,-1.125) {$\dots$};
\draw (-0.5,-2) to (0.5,-2);
\node[dot,label=below:{$M_r$}] (mr) at (0,-2.5) {};
\draw[dashed] (-0.25,-2.75) to (-1,-2.75);

\node[dot,label=below:{$L_c$}] (lc) at (-6,2.25) {};
\draw (-6.5,1.25) to (-5.5,1.25);
\draw (t3) to (lc) to (t1);
\node[dot,label=below:{$L_u^1$}] (lu1) at (-6,0.75) {};
\draw (lu1) to (t3);
\draw (-6.5,-0.25) to (-5.5,-0.25);
\node at (-6,-1.125) {$\dots$};
\draw (-6.5,-2) to (-5.5,-2);
\node[dot,label=below:{$L_r$}] (lr) at (-6,-2.5) {};
\draw (ku2) to (lr) to (ku3);
\draw[dashed] (-5.75,-2.25) to (-4.75,0.75);

\node[dot] (i1) at (1.5,1.75) {};
\node[dot] (i2) at (1.5,1.25) {};
\draw (i1) to (i2) to (mc) to (i1);
\node[dot] (i3) at (1.5,0.5) {};
\draw (mr) to (i3) to (mu1);
\draw[dashed,rounded corners] (1.5,2.5) to (1.5,3) to (-6,3) to (-6,2.5);
\end{tikzpicture}
\caption{
An illustration of the partition of $K$, $L$, and $M$.
Note that $P = V(K) \cup I$ and $N = M \cup L$. 
Dashed lines indicate some of the non-existing edges.
}\label{fig:partitionKandM}
\end{figure}
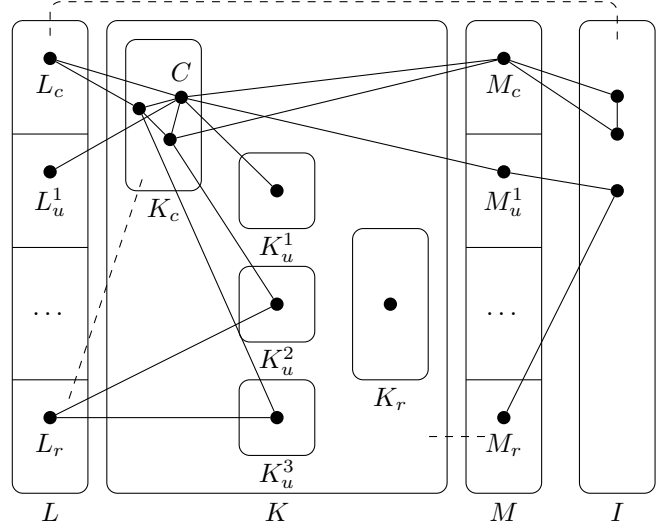

\subparagraph*{Handling $K_r$ and $L_r$}
Since $L \subseteq N$ is independent in $G$ and $G$ is connected, every vertex of $L$ has a neighbour in $K$.
If, in $G$, a vertex $v \in L_r$ has only neighbours in $K_u^i$ for one $i \in [3]$, then a 3-colouring of $G-v$ that extends $\psi$ can be extended to a 3-colouring of $G$ by assigning colour $i$ to $v$.
We remove any such $v$ from $G$ and continue.
Now, every vertex of $L_r$ has neighbours in $K_u^i$ for at least two distinct $i \in [3]$, or has a neighbour in $K_r$.
We prove two claims, one for each of the two described types of vertices in $L_r$.

\begin{figure}
\centering
\begin{tikzpicture}
\begin{scope}[shift={(-3,0)}]
\draw[rounded corners] (-2,-1) rectangle (-1,1);
\node at (-1.5,-1.25) {$K_c$};
\draw[rounded corners] (-0.5,-1) rectangle (0.5,1);
\node at (0,-1.25) {$K_u$};
\draw[rounded corners] (1,-1) rectangle (2,1);
\node at (1.5,-1.25) {$L_r$};

\node[dot,label=above:$u$] (u) at (1.5,0.5) {};
\node[dot,label=below:$v$] (v) at (1.5,-0.5) {};
\node[dot,label=above:$x$] (x) at (0,0.5) {};
\node[dot,label=below:$w$] (w) at (0,-0.5) {};
\draw (0.5,0) to (-0.5,0);
\foreach \i in {1, ..., 3} {
	\ifthenelse{\i=2}{
		\node[dot,shift={(-1.5,0)},label=below:$C$] (t\i) at (\i*120+45:0.333) {};
	}{
		\node[dot,shift={(-1.5,0)}] (t\i) at (\i*120+45:0.333) {};
	}
};
\draw (t1) to (t2) to (t3) to (t1);
\draw (t2) to (w) to (v);
\draw (t3) to (x) to (u);
\draw[dashed] (x) to (v);
\draw (x) to (w);
\end{scope}
\begin{scope}[shift={(3,0)}]
\draw[rounded corners] (-2,-1) rectangle (-1,1);
\node at (-1.5,-1.25) {$C'$};
\draw[rounded corners] (-0.5,-1) rectangle (0.5,1);
\node at (0,-1.25) {$K \setminus C'$};
\draw[rounded corners] (1,-1) rectangle (2,1);
\node at (1.5,-1.25) {$L_r'$};

\node[dot,label=below:$x$] (x) at (0,-0.5) {};
\foreach \i in {1, ..., 3} {
	\ifthenelse{\i=2}{
		\node[dot,shift={(-1.5,0)}] (t\i) at (\i*120+45:0.333) {};
	}{
		\ifthenelse{\i=3}{
			\node[dot,shift={(-1.5,0)},label=above:$u$] (t\i) at (\i*120+45:0.333) {};
		}{
			\node[dot,shift={(-1.5,0)},label=below:$z$] (t\i) at (\i*120+45:0.333) {};
		}
	}
};
\draw (t1) to (t2) to (t3) to (t1);
\node[dot,label=above:$v$] (v) at (0,0.5) {};
\node[dot,label=below:$y$] (y) at (1.5,0) {};
\draw (t3) to (v) to (x) to (y);
\draw[dashed] (x) to (-1.2,-0.2);
\end{scope}
\end{tikzpicture}
\caption{
Left: Proof of Claim~\ref{claim:Lr-Type1}. 
Right: Proof of Claim~\ref{claim:Lr-Type2}.
Dashed lines indicate non-existing edges. 
}\label{fig:claims-L}
\end{figure}
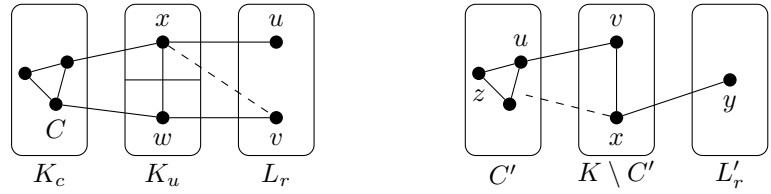

\begin{claim}\label{claim:Lr-Type1}
For $i,j \in [3]$ with $i\neq j$,
if $u \in L_r$ has a neighbour in $K_u^i$, and $v \in L_r$ has a neighbour in $K_u^j$,
then $u$ and $v$ have the same neighbours in $K_u^i \cup K_u^j$.
\end{claim}
\begin{proof}
Assume, for a contradiction, that $x \in K_u^i$ is a neighbour of $u$, but not a neighbour of $v$.
Let $w \in K_u^j$ be a neighbour of $v$.
Consider a shortest $x$-$w$-path $Q$ in $G$ with internal vertices in $K_c$.
As $Qv$ is not an induced $P_5$ in $G+F$, we must have $xw \in E(G)$.
Let $y$ be the neighbour of $x$ in $Q$, and let $z$ be a neighbour of $y$ that is adjacent to neither $x$ nor $w$.
Note that $z$ exists since every vertex of $K_c$ has two neighbours of two distinct colours.
Now, $vwxyz$ is an induced $P_5$ in $G+F$, a contradiction; see Figure~\ref{fig:claims-L} left.
\end{proof}

Let $L_r'$ be the set of vertices of $L_r$ with a neighbour in $K_r$.

\begin{claim}\label{claim:Lr-Type2}
A single vertex of $K$ dominates the vertices of $K_r \cup L_r'$.
\end{claim}
\begin{proof}
If $K_r = \emptyset$, then $L_r' = \emptyset$ and the statement is trivial. Hence, $K_r \not= \emptyset$. As $K$ is a connected $P_5$-free graph, $K$ contains a connected dominating set $D$ that induces a $P_3$-free graph or a $C_5$~\cite{camby2016newchar}. As we are in Case~2, $D$ cannot be a $C_5$. Hence, $D$ is a clique.

If $|D| \geq 4$, then $G$ contains a clique of order at least~$4$, which we already excluded. If $|D| = 3$, then $K$ contains a $C_3$ that dominates $K$. By the choice of $C$ and the fact that we are in Case~2, $C$ dominates $K$. Hence, our application of the Propagation Rule ensures that $K_r = \emptyset$, a contradiction. If $|D|=1$ and the vertex of $D$ is in $C$, then we arrive at a contradiction as before. If $|D|=1$ and the vertex of $D$ is not in $C$, then this vertex and $C$ form a clique of order at least~$4$, which we already excluded. It remains that $|D| = 2$. In other words, $K$ contains a dominating edge $uv$.

We must have that $N_K(u)$ and $N_K(v)$ are disjoint; otherwise, there would be a dominating triangle in $K$, which we can exclude as before. Without loss of generality, let $N_G(u)$ contain at least two vertices of $C$. This implies $u \in K_c$. As there is no edge between $K_c$ and $K_r$ by definition of $K_r$, $v$ dominates $K_r$.

It remains to show that $v$ is complete to $L_r'$. Suppose $y \in L_r' \setminus N_G(v)$ exists. Let $x \in K_r$ be a neighbour of $y$. As $u$ neighbours two vertices of $C$ and $u \in K_c$,  vertex $u$ is in a cycle $C'$ of length~$3$, 
which is contained in $K_c$ (possibly $C = C'$).
Let $z \in V(C') \setminus \{u\}$. As $N_K(u)$ and $N_K(v)$ are disjoint, $v$ is not adjacent to $z$. Also, $z$ is not adjacent to $x$ as $x \in K_r$ and $z \in K_c$, and $y$ is not adjacent to $u$, as $y \in L_r$ and $u \in K_c$. As $zuvxy$ is not an induced $P_5$ in $G+F$, we obtain $vy \in E(G)$, a contradiction; see Figure~\ref{fig:claims-L} right. Hence, $v$ dominates $L_r'$.
\end{proof}

Claim~\ref{claim:Lr-Type1} and Claim~\ref{claim:Lr-Type2} together imply that:

\begin{claim}\label{claim:KrLr}
$K_r \cup L_r$ is dominated by a set $D$ of at most two vertices of $K$.
\end{claim}

\subparagraph*{Handling $M_r$ and $J$}
We now describe the structure of the edges between $K$ and $M = N_G(I)$.

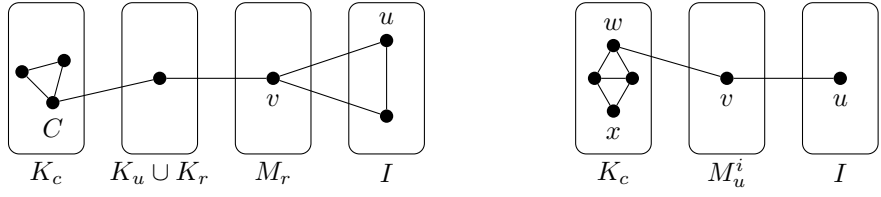
\begin{figure}
\centering
\begin{tikzpicture}
\begin{scope}[shift={(-3,0)}]
\draw[rounded corners] (-3.5,-1) rectangle (-2.5,1);
\node at (-3,-1.25) {$K_c$};
\draw[rounded corners] (-2,-1) rectangle (-1,1);
\node at (-1.5,-1.25) {$K_u \cup K_r$};
\draw[rounded corners] (-0.5,-1) rectangle (0.5,1);
\node at (0,-1.25) {$M_r$};
\draw[rounded corners] (1,-1) rectangle (2,1);
\node at (1.5,-1.25) {$I$};

\foreach \i in {1, ..., 3} {
	\ifthenelse{\i=2}{
		\node[dot,shift={(-3,0)},label=below:$C$] (t\i) at (\i*120+45:0.333) {};
	}{
		\node[dot,shift={(-3,0)}] (t\i) at (\i*120+45:0.333) {};
	}
};
\draw (t1) to (t2) to (t3) to (t1);

\node[dot] (nv) at (-1.5,0) {};
\node[dot,label=below:$v$] (v) at (0,0) {};
\node[dot,label=above:$u$] (u) at (1.5,0.5) {};
\node[dot] (nu) at (1.5,-0.5) {};
\draw (v) to (u) to (nu) to (v) to (nv) to (t2);
\end{scope}
\begin{scope}[shift={(3,0)}]
\draw[rounded corners] (-2,-1) rectangle (-1,1);
\node at (-1.5,-1.25) {$K_c$};
\draw[rounded corners] (-0.5,-1) rectangle (0.5,1);
\node at (0,-1.25) {$M_u^i$};
\draw[rounded corners] (1,-1) rectangle (2,1);
\node at (1.5,-1.25) {$I$};

\node[dot,label=below:$v$] (v) at (0,0) {};
\node[dot,label=below:$u$] (u) at (1.5,0) {};
\node[dot,label=above:$w$] (to) at (-1.5,0.433) {};
\node[dot] (tl) at (-1.75,0) {}; 
\node[dot] (tr) at (-1.25,0) {};
\node[dot,label=below:$x$] (tu) at (-1.5,-0.433) {};
\draw (tr) to (tl) to (to) to (tr) to (tu) to (tl);
\draw (to) to (v) to (u);
\end{scope}
\end{tikzpicture}
\caption{Left: Proof of Claim~\ref{claim:M-complete-Kci}~(i). Right: Proof of Claim~\ref{claim:M-complete-Kci}~(ii).}\label{fig:M-complete-Kci}
\end{figure}

\begin{claim}\label{claim:M-complete-Kci}
(i) Every vertex of $M_r$ has no neighbour in $K$, and (ii)
For every $i \in [3]$, $M_u^i$ is complete to $K_c^i$.
\end{claim}
\begin{proof}
We first prove (i).
Suppose, for the sake of contradiction, that there exists a vertex $v \in M_r$ that has a neighbour in $K$. Since $v \in M_r \subseteq M$, $v$ has a neighbour $u \in I$. By definition of $M_r$, $v$ has no neighbour in $K_c$. Let $Q$ be a shortest $v$-$C$-path in $G$ with internal vertices in $K$.
The path $Q$ must contain a vertex of $K_u^i$ for some $i \in [3]$ by assumption and therefore has length at least 2. Then there is an induced $P_5$ in $G+F$ with vertices in $\{u\} \cup V(Q) \cup V(C)$, a contradiction; see Figure~\ref{fig:M-complete-Kci} left.

We continue with (ii).
For some $i \in [3]$, let $v \in M_u^i$ such that $v$ is not complete to $K_c^i$. Since $v \in M$, $v$ has a neighbour $u \in I$. Since $v \in M_u^i$, $v$ has a neighbour in $K_c^i$. Let $x \in K_c^i$ be a non-neighbour of $v$.
Let $w \in K_c^i$ be a neighbour of $v$ that is closest to $x$ in $G[K_c]$.
Let $Q$ be a shortest $w$-$x$-path in $G[K_c]$, which exists since $G[K_c]$ is connected.
As $w,x \in K_c^i$, they are not adjacent. Thus, $Q$ has length at least 2, and $uvQ$ contains an induced $P_5$ in $G+F$, a contradiction; see Figure~\ref{fig:M-complete-Kci} right.
Hence, for every $i \in [3]$, $M_u^i$ is complete to $K_c^i$.
\end{proof}

We continue with two claims describing the structure of the edges between $I$ and $M_c \cup M_u$.

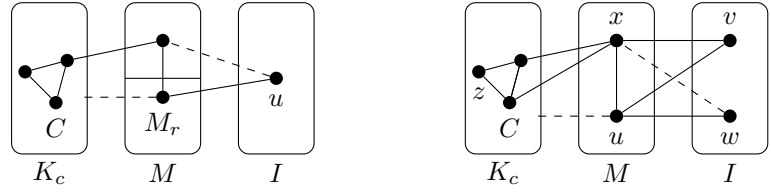
\begin{figure}
\centering
\begin{tikzpicture}
\begin{scope}[shift={(-3,0)}]
\draw[rounded corners] (-2,-1) rectangle (-1,1);
\node at (-1.5,-1.25) {$K_c$};
\draw[rounded corners] (-0.5,-1) rectangle (0.5,1);
\node at (0,-1.25) {$M$};
\draw[rounded corners] (1,-1) rectangle (2,1);
\node at (1.5,-1.25) {$I$};

\node[dot,label=below:$u$] (u) at (1.5,0) {};
\node[dot,label=below:{$M_r$}] (mr) at (0,-0.25) {};
\draw (0.5,0) to (-0.5,0);
\node[dot] (m) at (0,0.5) {};

\foreach \i in {1, ..., 3} {
	\ifthenelse{\i=2}{
		\node[dot,shift={(-1.5,0)},label=below:$C$] (t\i) at (\i*120+45:0.333) {};
	}{
		\node[dot,shift={(-1.5,0)}] (t\i) at (\i*120+45:0.333) {};
	}
};
\draw (t1) to (t2) to (t3) to (t1);
\draw (u) to (mr) to (m) to (t3);
\draw[dashed] (u) to (m);
\draw[dashed] (mr) to (-1.125,-0.25);
\end{scope}
\begin{scope}[shift={(3,0)}]
\draw[rounded corners] (-2,-1) rectangle (-1,1);
\node at (-1.5,-1.25) {$K_c$};
\draw[rounded corners] (-0.5,-1) rectangle (0.5,1);
\node at (0,-1.25) {$M$};
\draw[rounded corners] (1,-1) rectangle (2,1);
\node at (1.5,-1.25) {$I$};

\node[dot,label=above:$v$] (v) at (1.5,0.5) {};
\node[dot,label=below:$w$] (w) at (1.5,-0.5) {};
\node[dot,label=below:{$u$}] (u) at (0,-0.5) {};
\node[dot,label=above:{$x$}] (x) at (0,0.5) {};
\node[dot] (m) at (0,0.5) {};

\foreach \i in {1, ..., 3} {
	\ifthenelse{\i=2}{
		\node[dot,shift={(-1.5,0)},label=below:$C$] (t\i) at (\i*120+45:0.333) {};
	}{
		\ifthenelse{\i=1}{
			\node[dot,shift={(-1.5,0)},label=below:$z$] (t\i) at (\i*120+45:0.333) {};
		}{
			\node[dot,shift={(-1.5,0)}] (t\i) at (\i*120+45:0.333) {};
		}
	}
};
\draw (t1) to (t2) to (t3) to (t1);
\draw (u) to (v) to (x) to (t2) to (t3) to (x);
\draw (u) to (w);
\draw (x) to (u);
\draw[dashed] (u) to (-1.125,-0.5);
\draw[dashed] (w) to (x);
\end{scope}
\end{tikzpicture}
\caption{
Left: Proof of Claim~\ref{claim:I-nbor-McMu}.
Right: Proof of Claim~\ref{claim:same-nbors}.
Dashed lines indicate non-existing edges.
}\label{fig:claims-M}
\end{figure}

\begin{claim}\label{claim:I-nbor-McMu}
Every vertex of $I$ has a neighbour in $M_c \cup M_u$.
\end{claim}
\begin{proof}
Assume, for a contradiction, that the vertex $u \in I$ only has neighbours in $M_r$.
Since every vertex of $M_r$ has no neighbours in $K$ by Claim~\ref{claim:M-complete-Kci}~(i), a shortest $u$-$C$-path $Q$ in $G+F$ has length at least 3.
This implies that there is an induced $P_5$ in $G+F$ with vertices in $V(Q) \cup V(C)$, a contradiction; see Figure~\ref{fig:claims-M} left.
The claim follows.
\end{proof}

\begin{claim}\label{claim:same-nbors}
If $u \in M_r$, then every vertex of $N_G(u)$ has the same neighbours in $M_c \cup M_u$.
\end{claim}
\begin{proof}
Note that $N_G(u) \subseteq I$ by Claim~\ref{claim:M-complete-Kci}~(i) and since $M_r \subseteq N$ is independent.
Let $v, w \in N_G(u)$.
Assume, for a contradiction, that the vertex $x \in M_c \cup M_u$ is a neighbour of $v$, but not a neighbour of $w$.
By considering a shortest $u$-$C$-path in $G+F$ containing the vertices $v$ and $x$, we see that $ux \in F$.
Let $z \in K_c$ be a vertex that is not adjacent to $x$ in $G$, which exists, 
or $G$ would not be $3$-colourable.
Therefore, $wux$ together with a shortest $x$-$z$-path with internal vertices in $K_c$ contains an induced $P_5$ in $G+F$, a contradiction; see Figure~\ref{fig:claims-M} right.
As $v,w \in N_G(u)$ were arbitrary, the proof is complete.
\end{proof}

Claim~\ref{claim:G-Mr} is an important consequence of Claims~\ref{claim:I-nbor-McMu} and~\ref{claim:same-nbors}.

\begin{claim}\label{claim:G-Mr}
If there is 3-colouring $\psi'$ of $G - M_r$ that is an extension of $\psi$, then there is a 3-colouring of $G$ that is an extension of $\psi'$.
\end{claim}
\begin{proof}
Assume, for a contradiction, that for a vertex $u \in M_r$, there exist vertices $v_i \in N_G(u)$ with $\psi'(v_i) = i$ for every $i \in [3]$.
Note that $v_1,v_2,v_3 \in I$ by Claim~\ref{claim:M-complete-Kci}~(i).
By Claim~\ref{claim:I-nbor-McMu} and Claim~\ref{claim:same-nbors}, there exists a vertex $w \in M_c \cup M_u$ that is adjacent to $v_1$, $v_2$, and $v_3$, a contradiction to the fact that $\psi'$ is a $3$-colouring of $G-M_r$.
Therefore, for every vertex $u \in M_r$, there is a colour $i \in [3]$ such that no neighbour of $u$ in $G$ has colour $i$ under $\psi'$.
At this point,  choosing any such colour for every vertex of $M_r$ gives a 3-colouring of $G$ that is an extension of $\psi'$.
\end{proof}
\noindent
Claim~\ref{claim:G-Mr} implies that it suffices to decide if there is a 3-colouring $G-M_r$ that is an extension of $\psi$.
Hence, from now on, assume that $M_r = \emptyset$.
Recall that $J$ is the set of vertices of $I$ with no neighbour in $M_c$.
Consequently, by Claim~\ref{claim:I-nbor-McMu}, every vertex of $J$ has a neighbour in $M_u$.
Claim~\ref{claim:anti-complete} implies that every vertex of $M_c \cup M_u$ is either complete or anticomplete to each connected component of $G[I]$.
It follows that, if $u \in J$, then $J$ contains all vertices of the connected component of $u$ in $G[I]$.
We prove one more claim about the structure of the edges between $M_u$ and $J$.

\begin{figure}
\centering
\begin{tikzpicture}
\begin{scope}[shift={(0,0)}]
\draw[rounded corners] (-2,-1) rectangle (-1,1);
\node at (-1.5,-1.25) {$K_c$};
\draw[rounded corners] (-0.5,-1) rectangle (0.5,1);
\node at (0,-1.25) {$M_u$};
\draw[rounded corners] (1,-1) rectangle (2,1);
\node at (1.5,-1.25) {$J$};

\node[dot,label=below:$u$] (u) at (1.5,0) {};
\node[dot,label=below:{$v$}] (v) at (0,-0.5) {};
\node[dot,label=above:{$w$}] (w) at (0,0.5) {};

\foreach \i in {1, ..., 3} {
	\ifthenelse{\i=3}{
		\node[dot,shift={(-1.5,0)},label=above:$x$] (t\i) at (\i*120+45:0.333) {};
	}{
		\ifthenelse{\i=1}{
			\node[dot,shift={(-1.5,0)},label=above:$y$] (t\i) at (\i*120+45:0.333) {};
		}{
			\node[dot,shift={(-1.5,0)}] (t\i) at (\i*120+45:0.333) {};
		}
	}
};
\draw (t1) to (t2) to (t3) to (t1);
\draw (u) to (v) to (t2) to (t3) to (w) to (v);
\draw[dashed] (u) to (w);
\end{scope}
\end{tikzpicture}
\caption{
Proof of Claim~\ref{claim:Mu-J-complete}.
Dashed lines indicate non-existing edges.
}\label{fig:claim8}
\end{figure}
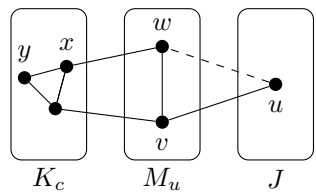

\begin{claim}\label{claim:Mu-J-complete}
If $M_u^i$ is nonempty for at least two $i \in [3]$, then the bipartite subgraph of $G$ spanned by the edges of $G$ with one end in $M_u$ and the other end in $J$ is complete.
\end{claim}
\begin{proof}
Let $K'$ be an arbitrary connected component of $G[J]$.
Keep in mind that $K'$ is a connected component of $G[I]$ too.
Let $i,j \in [3]$ with $i \neq j$ be such that $M_u^i$ is nonempty, and $K'$ has a neighbour $v$ in $M_u^j$.
Note that such $i$ and $j$ exist by assumption and Claim~\ref{claim:I-nbor-McMu}, and $v$ is complete to $K'$ by Claim~\ref{claim:anti-complete}.
We prove that $M_u^i$ is complete to $K'$.

Assume, for a contradiction, that $w \in M_u^i$ has no neighbour in $K'$.
Let $u$ be an arbitrary neighbour of $v$ in $K'$.
Consider a shortest $v$-$w$-path $Q$ with internal vertices in $K_c$, which exists since $G[K_c]$ is connected.
As $i \neq j$, the path $Q$ has length at least 3, and, by Claim~\ref{claim:M-complete-Kci}~(ii), the path $Q$ has length exactly 3.
Since $uQ$ is not an induced $P_5$ in $G+F$, we have $vw \in F$.
Let $x$ be the neighbour of $w$ in $Q$, let $k \in [3] \setminus \{i,j\}$, and let $y$ be a neighbour of $x$ in $K_c$ with colour $k$.
Note that $y$ exists since every vertex of $K_c$ has two neighbours in $K_c$ with two distinct colours.
Now $uvwxy$ is an induced $P_5$ in $G+F$, a contradiction; see Figure~\ref{fig:claim8}.
So $w$ has a neighbour in $K'$.
Claim~\ref{claim:anti-complete} implies that $w$ is complete to $K'$.
Since $w \in M_u^i$ was chosen arbitrarily, this proves that $M_u^i$ is complete to $K'$.

A similar argument shows that for $k \in [3] \setminus \{i,j\}$, if $M_u^k$ is nonempty, then $M_u^k$ is complete to $K'$ too.
By interchanging the roles of $i$ and $j$, we see that $M_u^j$ is complete to $K'$.
Since $K'$ was chosen arbitrarily, and since every such connected component of $G[J]$ has a neighbour in $M_u$ by Claim~\ref{claim:I-nbor-McMu}, this completes the proof.
\end{proof}

\subparagraph*{Colouring $G$}
At this point, we are in a position to decide if there is a $3$-colouring of $G$ that is an extension of $\psi$.
First, Claim~\ref{claim:KrLr} implies that $K_r \cup L_r$ is dominated by a set $D$ of at most two vertices of $K$. We branch on the (constantly many) consistent extensions of $\psi$ 
into $3$-colourings that additionally colour every vertex of $D$, which we call $\psi$ again for simplicity. 

Observe that every vertex of $K_r \cup L_r$ has a coloured neighbour now. As $M_r = \emptyset$, we now only need to achieve the same for $J$ in order to apply Lemma~\ref{lem:ext}.
If $J = \emptyset$, then Lemma~\ref{lem:ext} is directly applicable.
Therefore, we decide in polynomial time if there is a 3-colouring of $G$ that is an extension of $\psi$. If there is no such 3-colouring, then we backtrack.

We now assume that $J \not= \emptyset$. Since vertices in $J$ are not adjacent to $M_c$ by definition and $M_r = \emptyset$, $M_u \not= \emptyset$. If $M_u^i$ is nonempty for at least two $i \in [3]$, then we choose a vertex $v \in M_u$.
We branch on the extensions $\psi'$ of $\psi$ that additionally colour $v$.
Observe that now every vertex of $I$ has a coloured neighbour under $\psi'$ by the definition of $J$ and Claim~\ref{claim:Mu-J-complete}.
Now, Lemma~\ref{lem:ext} is applicable.
Therefore, we decide in polynomial time if there is a 3-colouring of $G$ that is an extension of $\psi'$.
If there is no such 3-colouring, then we backtrack.

It remains the case that there is exactly one $i \in [3]$ such that $M_u^i$ is nonempty. Every vertex in $J$ has neighbours only in $J \cup M_u^i$. In particular, for each connected component $K'$ of $G[J]$, which is a connected component of $G[I]$, the colour~$i$ may be used without creating conflicts outside of $K'$. Recall that $K'$ is bipartite by Claim~\ref{claim:isbipartite}. Hence, we wish to extend $\psi$ by, for each connected component $K'$ of $G[J]$, colouring one of its partite set by colour $i$. However, we cannot immediately decide which partite set, and make a small detour.

Let $K'$ be a connected component of $G[J]$ that contains an edge. Let $u$ be a neighbour of $K'$ in $M_u^i$. Since $u \in M_u^i$, it is adjacent to $K$, and thus neither complete nor anticomplete to $K$. Hence, $u$ is complete to $K'$ by Claim~\ref{claim:anti-complete}. Thus, $N_G(K')$ is complete to $K'$. Therefore, all vertices of $N_G(K')$ must receive the same colour in any $3$-colouring of $G$ that extends $\psi$. We ensure this first, for each such connected component $K'$, and then extend the colouring to $J$.

We apply the formula $\mathcal{F}$ of Lemma~\ref{lem:ext} to $G-J$, adapted as follows. For every connected component $K'$ in $G[J]$ that contains an edge, and for every two distinct vertices $u,v \in N_G(K')$, we add the clauses $(\bar{x}_u^k \vee x_v^k) \wedge (x_u^k \vee \bar{x}_v^k)$ to $\mathcal{F}$ for every $k \in [3] \setminus \{i\}$. These clauses ensure that two such vertices $u$ and $v$ receive the same colour. (Alternatively we could identify these vertices. At this point it does not matter that this does not preserve probe $P_5$-freeness.) After that, we resolve the satisfiability of the 2-SAT formula $\mathcal{F}$ in polynomial time~\cite{aspvall1979lineartime}. If $\mathcal{F}$ is not satisfiable, then there is no 3-colouring of $G$ that is an extension of $\psi$, and we backtrack. Otherwise, let $\psi'$ be a 3-colouring of $G-J$ obtained from a satisfying assignment of $\mathcal{F}$. We can extend $\psi'$ to a $3$-colouring of $G$ by assigning colour $i$ to isolated vertices in $G[J]$, and by assigning the remaining two colours to the nontrivial bipartite connected components of $G[J]$, which is possible due to the extra clauses we added to $\mathcal{F}$.
This completes the proof of Theorem~\ref{t-poly}.
\end{proof}

\section{The Proof of the NP-Completeness Part of Theorem~\ref{thm:3col-probe-p5}}\label{sec:probe-p6} 

Theorem~\ref{thm:3col-probe-p5} states that 
for  $t\geq 1$, \ProbekCol{$3$} on partitioned probe $P_t$-free graphs is polynomial-time solvable if $t\leq 5$ and \NP-complete if $t\geq 6$.
In Section~\ref{sec:probe-p5} we showed the polynomial part of Theorem~\ref{thm:3col-probe-p5}. We now show the \NP-completeness part by reducing from \PreExt{}, which has as input: an integer $k \geq 3$, a graph $G$ with at least $k$ vertices, and a partial $k$-colouring $\psi$ of $G$ that assigns $k$ vertices $v_1,\dots,v_k$ colours $1,\dots,k$, respectively. Can $\psi$ be extended to a $k$-colouring of $G$?
Bodlaender et al.~\cite{bodlaender1994scheduling} proved that this problem is \NP-complete, even if $k=3$, $G$ is bipartite 
and the precoloured vertices all belong to the same partition set of $G$. 
Cai~\cite{Ca03} used this result to prove that $3$-{\sc Colouring} is \NP-complete for graphs that become bipartite by deleting three edges.
We use the gadget of~\cite{Ca03} to show the following, which implies the \NP-completeness part of Theorem~\ref{thm:3col-probe-p5}.

\begin{theorem}\label{thm:3col-probe-p6}
\kCol{$3$} is \NP-complete on partitioned probe $(P_6,3P_2,2P_3)$-free graphs.
\end{theorem}

\begin{proof}
We reduce an instance $(3,G,\psi,\{v_1,v_2,v_3\})$ of \PreExt{}, where $G$ is a bipartite graph with bipartition $A$ and $B$, and the precoloured vertices $v_1,v_2,v_3$ belong to $A$ without loss of generality, to an instance of \ProbekCol{$3$}.
As mentioned, this variant of \PreExt{} is still \NP-complete~\cite{bodlaender1994scheduling}.
The bipartition of $G$ can be computed in polynomial time.
Let $G'$ be the graph from~\cite{Ca03}, which is obtained in polynomial-time from $G$ by turning $\{v_1,v_2,v_3\}$ into a clique.
The graph $G'$ is probe $(P_6, 2P_3, 3P_2)$-free, which is witnessed by the fact that the graph obtained from $G'$ by turning the independent set $B$ into a clique is $(P_6, 2P_3, 3P_2)$-free.
It is easy to see that $(3,G,\psi,\{v_1,v_2,v_3\})$ is a yes-instance of \PreExt{} if and only if $(G', A, B)$ is a yes-instance of \ProbekCol{$3$}.
This proves that \ProbekCol{$3$} is \NP-hard on partitioned probe $(P_6,2P_3,3P_2)$-free graphs.
\end{proof}

\section{Additional Results and Concluding Remarks}\label{s-con}

In our paper, we considered the probe graph model introduced by Zhang et al.~\cite{zhang1994analgorithm}. Our aim was to research {\it to what extent} polynomial-time results for $H$-free graphs can be extended to probe $H$-free graphs.
We first gave a dichotomy for \Col{} restricted to (partitioned) probe $H$-free graphs and then showed our main result, which states that the known polynomial-time result for \kCol{$3$} for $P_5$-free graphs (whose yes-instances all have bounded mim-width) can be extended to partitioned probe $P_5$-free graphs (whose yes-instances even have unbounded sim-width). We also proved that this result cannot be generalized to partitioned probe $P_6$-free graphs unless $\P = \NP$ by showing \NP-completeness even for partitioned $(P_6,3P_2,2P_3)$-free graphs.
As \kCol{$3$} is polynomial-time solvable even for $P_7$-free graphs~\cite{BCMSSZ18} and $sP_2$-free graphs for all $s\geq 1$~\cite{DLRR12},
our results give a clear indication of the difference in computational complexity if not all edges of the input graph are known, under the probe graph model.

Our results also lead to a range of natural directions for future work. We sketch three directions below, before concluding with three broader questions.

\subsection{Towards a Dichotomy for \kCol{$3$} on Partitioned Probe $H$-Free Graphs}
The dichotomy for \kCol{$3$} for partitioned probe $H$-free graphs has not been fully settled.  We are able to prove the following additional result:

\begin{theorem}\label{thm:3col-probe-sp1p3}
For every $s \geq 0$, \ProbekCol{$3$} is polynomial-time solvable on partitioned probe $(P_3+sP_1)$-free graphs.
\end{theorem}

\begin{proof}
Let $(G,P,N)$ be a partitioned probe $(P_3+sP_1)$-free graph. Let $F \subseteq {N \choose 2}$ be such that $G + F$ is $(P_3+sP_1)$-free.
We define $F$ only for verifying correctness; the polynomial-time algorithm does not use $F$.
We may assume that $G$ is connected; otherwise, we run the algorithm on each connected component. We verify in polynomial time that $G$ has no clique on $4$ vertices; otherwise, $G$ is not $3$-colourable. We verify in polynomial time that each vertex of $N$ has degree at least $3$; otherwise, we can remove such a vertex $v$ and run the algorithm on $G-v$, as there is always a free colour for $v$ if $G-v$ is $3$-colourable. We may assume that $N$ has at least one vertex; otherwise, we just solve \kCol{$3$} in polynomial time~\cite{HKLSS10}, as $G$ would be $(P_3+sP1)$-free.

We distinguish two cases, depending on whether $G[P]$ has an induced $P_3$. This can be checked in polynomial time.

\subparagraph*{Case 1: $G[P]$ has no induced $P_3$.} We need the following claim:

\begin{claim} \label{clm:p3:nop3}
If $G$ is $3$-colourable, then for every $u \in N$ it holds that:
\begin{itemize}
\item[(i)] $u$ has at most $3(s+2)$ non-neighbours in $P$;
\item [(ii)] there exists a subset $S$ of the non-neighbours of $u$ in $P$ such that $S$ is independent and $G-(S \cup (N \setminus N_G(S)))$ is bipartite.
\end{itemize}
\end{claim}
\begin{proof}
Let $\psi$ be a $3$-colouring of $G$. We start by proving (i). Since $G[P]$ is $P_3$-free and $3$-colourable by assumption, $G[P]$ is a disjoint union of cliques of size at most~$3$. If $u$ has neighbours in at most one connected component of $G[P]$, then since $u$ has degree at least~$3$, that connected component has size at least~$3$ and $G$ contains a clique on $4$ vertices, a contradiction. Hence, $u$ has neighbours in at least two connected components of $G[P]$. Let $v,w$ be neighbours of $u$ in distinct connected components of $G[P]$. Suppose for sake of contradiction that $u$ has more than $3(s+2)$ non-neighbours in $P$. Since each connected component of $G[P]$ has size at most~$3$, $u$ has non-neighbours in at least $s+2$ distinct connected components of $G[P]$. Hence, $u$ has $s$ non-neighbours in distinct connected components that do not contain $v$ or $w$. Hence, $G+F$ contains an induced $P_3+sP_1$, a contradiction. It follows that $u$ has at most $3(s+2)$ non-neighbours in $P$.

For (ii), without loss of generality, assume that $\psi(u)=3$. Let $S = \psi^{-1}(3) \cap P$; clearly, $S$ is a subset of the non-neighbours of $u$ in $P$ and $S$ is independent. Moreover, for every $v \in P \setminus S$, $\psi(v) \in \{1,2\}$. Also, for every $w \in N_G(S) \cap N$, $\psi(w) \in \{1,2\}$. Hence, $G-(S \cup (N \setminus N_G(S)))$ is bipartite.
\end{proof}

We are now ready for the algorithm. Let $u \in N$, which exists as $N$ is nonempty. If $u$ has more than $3(s+2)$ non-neighbours in $P$, then return that $G$ is not $3$-colourable. This is correct by Claim~\ref{clm:p3:nop3}(i). Branch on each subset $S$ of the non-neighbours of $u$ in $P$. If $S$ is not an independent set or $G-(S \cup (N \setminus N_G(S)))$ is not bipartite, reject the branch; otherwise, accept it. The branching algorithm takes polynomial time, since $|S| \leq 3(s+2)$. If there is an accepted branch, then clearly, using the $2$-colouring of $G-(S \cup (N \setminus N_G(S)))$ plus assigning colour $3$ to $S \cup (N \setminus N_G(S))$ is a $3$-colouring of $G$. By Claim~\ref{clm:p3:nop3}(ii), there is an accepted branch if $G$ is $3$-colourable. Hence, the algorithm is correct and runs in polynomial time.

\subparagraph*{Case 2: $G[P]$ contains an induced $P_3$.} We describe the algorithm. We distinguish two cases, depending on $|P|$.

\subparagraph*{Case 2a: $|P| \leq 4s+1$.} Branch on each $3$-colouring $\psi'$ of $P$. There are constantly many such branches, as $|P| \leq 4s+1$ and $s$ is fixed. Since $G$ is connected, every vertex of $N$ is adjacent to some vertex of $P$. Hence, Lemma~\ref{lem:ext} is directly applicable. Therefore, we decide in polynomial time if there is a $3$-colouring of $G$ that is an extension of $\psi'$. If there is no such $3$-colouring, then we backtrack. This part of the algorithm is clearly correct and runs in polynomial time.

\subparagraph*{Case 2b: $|P| > 4s+1$.} In polynomial time, find an induced subgraph $Q$ of $G[P]$ isomorphic to $P_3$. In polynomial time, find a maximal independent set $I$ of $G[P] \setminus N_G[V(Q)]$. If $|I| \geq s$, then $I \cup V(Q)$ induces a $P_3+sP_1$ in $G[P]$ and thus in $G+F$, a contradiction. Hence, $|I| \leq s-1$. Let $D = V(Q) \cup I$. Clearly, every vertex of $P$ is dominated by $D$ and $|D| \leq s+2$.

Branch on all disjoint subsets $S$ of $P \setminus D$ such that $|S| \leq 3s$. There are polynomially many such branches, as $s$ is fixed. If there is a $u \in N$ such that $u$ is not adjacent to $D \cup S$, then backtrack. Otherwise, branch on each $3$-colouring $\psi'$ of $D \cup S$. There are constantly many such branches, as $|D \cup S| \leq 4s+2$ and $s$ is fixed. By assumption and construction, every vertex of $P \cup N$ is either in or adjacent to some vertex of $D \cup S$. Hence, Lemma~\ref{lem:ext} is directly applicable. Therefore, we decide in polynomial time if there is a $3$-colouring of $G$ that is an extension of $\psi'$. If there is no such $3$-colouring, then we backtrack.

We now show that this part of the algorithm is correct. If there is a $3$-colouring $\psi$ of $G$, then for every $i \in [3]$, pick any set $S_i \subseteq \psi^{-1}(i) \cap (P \setminus D)$ such that $|S_i| = \min\{|\psi^{-1}(i) \cap (P \setminus D)|, s\}$. Consider $S = S_1 \cup S_2 \cup S_3$. Since $|S_1 \cup S_2 \cup S_3| \leq 3s$, this set $S$ will be considered by the algorithm.

Let $u \in N$. We claim that $u$ has a neighbour in $D \cup S$. Suppose not. Since $u$ has degree at least~$3$, there is a colour $i \in [3]$ such that $u$ has at least two neighbours of colour~$i$. Since $u \in N$, these neighbours are in $P$, and any two of them together with $u$ induce a $P_3$. If $u$ has no neighbours in $D \cup S$, then it has no neighbours in $S_i \cup (\psi^{-1}(i) \cap D)$ in particular. Since $u$ has neighbours in $\psi^{-1}(i) \cap P$, it holds that $S_i \cup (\psi^{-1}(i) \cap D) \subset \psi^{-1}(i) \cap P$. Hence, $|S_i| \geq s$ by the choice of $S_i$. Thus, $u$ has at least $s$ non-neighbours in $\psi^{-1}(i) \cap P$. Any $s$ of them, together with the $P_3$, yields an induced $P_3+sP_1$ in $G+F$, a contradiction. Hence, $u$ has a neighbour in $D \cup S$.

It follows that for this choice of $S$, the algorithm will not backtrack. Then for the $3$-colouring $\psi'$ of $D \cup S$ that is the restriction of $\psi$ to $D \cup S$, the algorithm will succeed to find a $3$-colouring by Lemma~\ref{lem:ext}.
\end{proof}

\noindent
Theorem~\ref{thm:3col-probe-p5}, Theorems~\ref{thm:3col-probe-p6}--\ref{thm:3col-probe-sp1p3} and the result that \kCol{$3$} is \NP-complete on $H$-free graphs if $H$ is not a linear forest~\cite{EHK98,holyer1981thenpcompleteness} leave only the following open cases:

\begin{open}
Determine the complexity of \ProbekCol{3} on partitioned probe $H$-free graphs when $H$ is $2P_2+sP_1$ ($s \geq 1$), $P_3+P_2+sP_1$ ($s \geq 0$), $P_4+sP_1$ ($s \geq 1$), $P_4+P_2+sP_1$ ($s \geq 0$), or $P_5+sP_1$ ($s \geq 1$).
\end{open}

\subsection{Solving \kCol{$k$} on Probe $H$-Free Graphs for $k\geq 4$}

While this paper has mostly focussed on \kCol{$3$}, looking more generally at \kCol{$k$} gives rise to interesting problems and observations. Since \kCol{$k$} is polynomial on $P_5$-free graphs~\cite{HKLSS10} even for all $k\geq 3$, we ask:

\begin{open}\label{o-k4}
For $k \geq 4$, determine the complexity of \ProbekCol{$k$} on partitioned probe $P_5$-free graphs.
\end{open}

\noindent
Crucial properties in our proof for \ProbekCol{$3$} on partitioned probe $P_5$-free graphs, such as the fact that there is a single non-bipartite connected component and that no vertex is complete to the cycle~$C$ we pick in it, no longer hold if $k\geq 4$. We do note that probe $(K_s,P_5)$-free graphs have bounded mim-width for every $s\geq 1$, due to $(K_s,P_5)$-free graphs having bounded mim-width for every $k\geq 1$~\cite{brettel2022listkcolouring} and Proposition~\ref{prop:probe-cw-mimw}. Hence, as we may assume that an input graph for {\sc $4$-Colouring} is $K_5$-free,
a good starting point is to consider {\sc $4$-Colouring} for $K_5$-free partitioned probe $P_5$-free graphs, or even for $K_5$-free partitioned probe $2P_2$-free graphs. 

We recall that for solving $3$-{\sc Colouring} on probe $P_5$-free graphs, we only need the partition into $P$ and $N$.
To solve Problem~\ref{o-k4}, it would also be interesting to research if the problem becomes easier under the assumption that we also know the set of edges $F$.

As another starting point for solving Problem~\ref{o-k4}, we can show the following result:

\begin{theorem}\label{thm:p2sp1}
For every $s \geq 0$ and $k \geq 1$, \kCol{$k$} is polynomial-time solvable on (not necessarily partitioned) probe $(P_2+sP_1)$-free graphs.
\end{theorem}

\begin{proof}
We first show that every probe $(P_2+sP_1)$-free graph is $(s+1)P_2$-free.
Let $(G,P,N)$ be a partitioned probe $(P_2+sP_1)$-free graph. Let $F \subseteq {N \choose 2}$ be such that $G + F$ is $(P_2+sP_1)$-free. Suppose $G$ has an induced subgraph $H$ isomorphic to $(s+1)P_2$. If no vertices of $H$ are in $N$, then $H$ is contained in $G[P]$, and thus $G+F$ has an induced subgraph isomorphic to $P_2+sP_1$, a contradiction. Hence, at least one vertex of $H$ is in $N$. As $N$ is independent in $G$, it holds that if $u \in V(H) \cap N$, then the neighbour of $u$ in $H$ must be in $P$. Hence, for each edge of $H$, at least one endpoint is in $P$. Combined, this implies that $G+F$ has an induced $P_2+sP_1$, a contradiction.
As \kCol{$k$} can be solved in polynomial time on $sP_2$-free graphs for all $k \geq 1$ (see e.g.~\cite[in Theorem~5]{DLRR12} or~\cite[Theorem~6]{GJPS17}), the result follows.
\end{proof}

\noindent
We note that for $k=3$, Theorem~\ref{thm:p2sp1} does not need the partition $V=P\cup N$, whereas Theorem~\ref{thm:3col-probe-sp1p3} does, so the two theorems are not comparable.

\subsection{Towards Recognizing Probe $H$-Free Graphs}
The results in this paper give nontrivial insights into the structure of probe $H$-free graphs in general, and probe $P_5$-free graphs in particular. Several interesting challenges remain. First, we note that we could prove Theorem~\ref{thm:p2sp1} by showing that for all $s\geq 1$, every probe $(P_2+sP_1)$-free graphs is $(s+1)P_2$-free, and we ask:

\begin{open}\label{o-p26}
Are there  other sets~${\cal H}$, for which there exists a {\it finite} set ${\cal H'}$ such that every probe ${\cal H}$-free graph is ${\cal H'}$-free? 
\end{open}

\noindent
An inclusion as in Problem~\ref{o-p26} is in general strict, e.g., probe $P_5$-free graphs form a hereditary graph class that is not finitely defined. To explain this, as probe $H$-free graphs are closed under vertex deletion, there exists a unique minimal set of graphs ${\cal F}_H$ such that a graph is probe $H$-free if and only if it is ${\cal F}_H$-free. We can show that ${\cal F}_{P_5}$ is infinite as follows.

We first need several auxiliary results. We recall that every even cycle is bipartite. Hence, we can let $N$ be one of the two independent sets of the bipartition and change $N$ into a clique to find 
that every even cycle $C_{2\ell}$ ($\ell\geq 2$) is probe $2P_2$-free and thus probe $P_5$-free, just like the $C_3$ and $C_5$, both of which are even $2P_2$-free. However, for odd cycles of length at least~$7$, this is no longer true:

\begin{lemma}\label{l-inf}
For every $\ell\geq 3$, the cycle $C_{2\ell+1}$ is not probe $P_5$-free (but is probe $P_6$-free).
\end{lemma}

\begin{proof}
Let $\ell\geq 3$. For a contradiction, assume that $G=C_{2\ell+1}$ is probe $P_5$-free. Then there exists an independent set $N$ in $G$ such that $G+F$ is $P_5$-free where $F$ is some set of new edges that we add between vertices of $N$. Since $\ell\geq 3$, we find that $G$ is not $P_5$-free. Hence, $N$ is not empty. As before, let $P=V(G)\setminus N$.

As $N$ is an independent set and $G$ is not bipartite, $P$ is not an independent set. As $N$ is not empty and $G$ is a cycle, this means that one of the connected components of $G[P]$ is a path $R=u_1\cdots u_r$ for some $r\geq 2$. By definition, no edge of $F$ is incident to a vertex of $R$. 
As $G+F$ is $P_5$-free, this means that $r\leq 4$. 

Let $u_{r+1}$ be the neighbour of $u_r$ on $G$ that does not belong to $R$, so $u_{r+1}\in N$.
Let $u_{r+2}$ be the neighbour of $u_{r+1}$ on $G$ that does not belong to $R$. As $u_{r+1}$ belongs to $N$ and $N$ is an independent set, $u_{r+2}$ belongs to $P$. Let  $R'=u_1\cdots u_{r+2}$. Note that $R'$ is an induced path in $G$, as $\ell\geq 3$ and $r\leq 4$. Moreover, as $F$ contains no edges incident to a vertex of $R$, we find that $R'$ is also an induced path in $G+F$. As $G+F$ is $P_5$-free, this means that $r=2$.

Let $u_5$ be the neighbour of $u_4$ not equal to $u_3$. As $\ell\geq 3$, we find that $u_1u_2u_3u_4u_5$ is an induced $P_5$ in $G$. If $u_5\in P$, then $u_1u_2u_3u_4u_5$ is also an induced $P_5$ in $G+F$. Hence, $u_5$ belongs to $N$. By the same arguments, $F$ must contain the edge~$u_3u_5$. 

Let $u_6$ be the neighbour of $u_5$ not equal to $u_4$. As $u_5\in N$ and $N$ is an independent set, $u_6$ belongs to $P$. As $\ell\geq 3$, the path $u_1u_2u_3u_5u_6$ is an induced $P_5$ in $G$. Again, as $F$ contains no edge incident to a vertex of $P$, we find that $u_1u_2u_3u_5u_6$ is even an induced $P_5$ in $G+F$, a contradiction.

Finally, to show that $G$ is probe $P_6$-free, we let $N$ be a maximal independent set in $G$ and change $N$ into a clique to obtain a graph $G'$. We note that $G[P]$ is the disjoint union of an edge and a set of isolated vertices. Hence, every induced path in $G'$ has at most five vertices.
\end{proof}

\noindent
We use Lemma~\ref{l-inf} to prove the following:

\begin{lemma}\label{l-2}
Probe $P_5$-subgraph-free graphs form a proper subclass of $(C_7,C_9,C_{11},\ldots)$-free graphs.
\end{lemma}

\begin{proof}
By Lemma~\ref{l-inf}, every probe $P_5$-subgraph-free graph is $(C_7,C_9,C_{11},\ldots)$-free. 

To show that the inclusion is strict, we construct the following graph. Take two vertex disjoint paths $u_1u_2u_3u_4u_5$ and $u_1'u_2'u_3'u_4'u_5'$ on five vertices. For $i=1,\ldots,5$, we add the edge 
$u_iu_i'$. For $i=1,\ldots,4$, we add also the edges $u_iu_{i+1}'$ and $u_i'u_{i+1}$. Let $L$ be the resulting graph. We note that for $i=1,\ldots,5$, the vertices $u_i$ and $u_i'$ are true twins in $L$. 

For a contradiction, assume that $L$ is probe $P_5$-free, implying $L$ contains an independent set $N$ such that $L+F$ is $P_5$-free for some subset of edges that can be added between vertices of $N$. Let $P=V(L)\setminus N$. Because each two vertices $u_i$ and $u_i'$ form a pair of true twins, $N$ contains at most one vertex of $\{u_i,u_i'\}$. By symmetry, we may assume that none of $u_1,\ldots u_5$ belongs to $N$. Hence, $u_1u_2u_3u_4u_5$ forms an induced $P_5$ in $L[P]$ and thus in $L+F$, a contradiction. We conclude that $L$ is not probe $P_5$-free.
It remains to observe that $L$ is $(C_7,C_9,C_{11},\ldots)$-free.
\end{proof}

\noindent
We can now conclude the following:

\begin{theorem} \label{thm:infinite}
It holds that $\{C_7,C_9,C_{11},\ldots\}\subsetneq {\cal F}_{2P_2}$ and $\{C_7,C_9,C_{11},\ldots\}\subsetneq {\cal F}_{P_5}$.
\end{theorem}
\begin{proof}
We first note that every disjoint union of paths is a bipartite graph. We recall that bipartite graphs are probe $2P_2$-free: let $N$ be one of the two partition classes and change $N$ into a clique by adding all the edges between vertices of $N$. 
Hence, every induced subgraph of every cycle is probe $2P_2$-free and thus also probe $P_5$-free.
It remains to combine this observation with Lemma~\ref{l-2} and the fact that no cycle $C_r$ is an induced subgraph of a cycle $C_s$ whenever $s\neq r$.
\end{proof}

\noindent
Theorem~\ref{thm:infinite} implies that ${\cal F}_{2P_2}$ and ${\cal F}_{P_5}$ are infinite. However, we do not know its exact composition:

\begin{open} \label{o-fp5}
Determine ${\cal F}_{P_5}$.
\end{open}

Lemma~\ref{l-2} showed that probe $P_5$-free graphs form a proper subclass of $(C_7,C_9,C_{11},\ldots)$-free graphs, that is, graphs with no odd hole of length at least~$7$.
This leads to the following natural question:

\begin{open}
Determine the complexity of $3$-{\sc Colouring} for $(C_7,C_9,C_{11},\ldots)$-free graphs.
\end{open}

\noindent
It is known that {\sc $3$-Colouring} is polynomial-time solvable for odd-hole free graphs, i.e., $(C_5,C_7,C_9,\ldots)$-free graphs (see~\cite{O})
and $(K_4,C_7,C_9,C_{11},\ldots)$-free graphs are $\chi$-bounded~\cite{CSSS20}.

We next consider the question of recognizing probe $H$-free graphs. It is known that probe $P_4$-free graphs can be recognized in polynomial time~\cite{chang2005ontherecognition}. However, we do not know whether this holds for probe $P_5$-free graphs and thus state:

\begin{open} \label{o-recog}
Determine the complexity of recognizing probe $P_5$-free graphs.
\end{open}

\noindent
An answer to Problem~\ref{o-fp5} could help solve Problem~\ref{o-recog}. If Problem~\ref{o-recog} is resolved positively, in the sense that a polynomial-time algorithm exists, then combined with Theorem~\ref{t-poly} this immediately implies that \kCol{$3$} can be solved in polynomial time on probe $P_5$-free graphs without a given partition $(P,N)$ of their vertex set.

\subsection{Broader Future Research Directions}
We finish with three other, {\it broader} directions for future work.
First, we recall our result (Proposition~\ref{p-easy}) that $C_3$-free probe $P_5$-free graphs (and thus all probe $(C_3,P_5)$-free graphs) are $3$-colourable, which generalizes the same result for $(C_3,P_5)$-free graphs~\cite{WS01}. There is a long history for obtaining constant (tight) bounds on the chromatic number of $(H_1,H_2)$-free graphs; see also the recent survey~\cite{CK25}. In particular, for all $r,t\geq 1$, every $(K_r,P_t)$-free graph can be coloured with at most $(t-2)^{r-2}$ colours~\cite{GHM03}, improving an older result in~\cite{Gy87}. 
It is also known e.g. that every $(C_3,sP_2)$-free graph is $(2s-2)$-colourable for all $s\geq 3$~\cite{Br02}, and 
 that every $(C_3,2P_3)$-free graph~\cite{Py13}, every $(C_3,P_2+P_4)$-free graph~\cite{broersma2012updating} and every $(K_4,2P_2)$-free graph~\cite{GH19} is $4$-colourable. In particular, subclasses of $P_5$-free graphs are well studied; see~\cite{CK23,DXX22,DXX23} for several of such results. 

\begin{quote}
{\it To what extent can all these results be extended to the probe version?}
\end{quote}

\noindent
Proposition~\ref{p-easy} illustrates there exist graph classes with a positive answer.

Second, generalizing the above question, we ask for which other graph classes $\mathcal{G}$, are \ProbeCol{} and \kCol{$k$} polynomially solvable on the class of (partitioned) probe graphs $\mathcal{G}_p$? Recall from Section~\ref{sec:intro} that \Col{} is polynomial-time solvable for probe chordal graphs (as these graphs are perfect) and that in 2012, Chandler et al.~\cite{chandler2012probegraphclasses} conjectured the same for partitioned probe perfect graphs.

Third, our results indicate that the reason for polynomial-time solvability of $3$-{\sc Colouring} for $P_5$-free graphs goes beyond boundedness of mim-width.
The question, which may also be asked for other problems, is if we can push further in this direction. So far,
we assumed that the set of non-probes $N$ is an independent set. This is an extreme case in a more general model in which we assume that we {\it do} know some of the edges in $G[N]$.

\begin{open}
 {\it Determine if $3$-{\sc Colouring} is polynomially solvable on graphs $(G,P,N)$ that can be made $P_5$-free by adding 
 new edges to $G[N]$, which may have some initial edges.} 
 \end{open}
 
 \noindent
 If we allow an arbitrary initial set of edges in $G[N]$, then the problem is already \NP-complete on graphs that can be made $2P_1$-free:
take the graph $G=(V,E)$ used in any \NP-hardness reduction for $3$-{\sc Colouring}; set $N:=V$; and let $F$ be the set of all edges that are not already in $G$ (so $G+F$ is complete, that is, $G+F$ is $2P_1$-free).

\subparagraph*{Acknowledgment}
We thank Cl\'ement Dallard and Andrea Munaro for helpful discussions and ideas that ultimately led to the proof of Theorem~\ref{thm:3col-probe-p5}.

\bibliography{references}

@article{Py13,
  author       = {Artem V. Pyatkin},
  title        = {Triangle-free $2{P}_3$-free
 graphs are $4$-colorable},
  journal      = {Discret. Math.},
  volume       = {313},
  number       = {5},
  pages        = {715--720},
  year         = {2013},
  url          = {https://doi.org/10.1016/j.disc.2012.10.019},
  doi          = {10.1016/J.DISC.2012.10.019},
  bibsource    = {dblp computer science bibliography, https://dblp.org}
}

@article{O,
author={Pascal Ochem},
url ={https://www.graphclasses.org/classes/refs1700.html\#ref\_1744}
}

@article{CSSS20,
  author       = {Maria Chudnovsky and
                  Alex Scott and
                  Paul D. Seymour and
                  Sophie Spirkl},
  title        = {Induced subgraphs of graphs with large chromatic number. {VIII.} Long
                  odd holes},
  journal      = {J. Comb. Theory, Ser. {B}},
  volume       = {140},
  pages        = {84--97},
  year         = {2020},
  url          = {https://doi.org/10.1016/j.jctb.2019.05.001},
  doi          = {10.1016/J.JCTB.2019.05.001},
  bibsource    = {dblp computer science bibliography, https://dblp.org}
}

@article{Ca03,
  author       = {Leizhen Cai},
  title        = {Parameterized Complexity of Vertex Colouring},
  journal      = {Discret. Appl. Math.},
  volume       = {127},
  number       = {3},
  pages        = {415--429},
  year         = {2003},
  url          = {https://doi.org/10.1016/S0166-218X(02)00242-1},
  doi          = {10.1016/S0166-218X(02)00242-1},
  bibsource    = {dblp computer science bibliography, https://dblp.org}
}

@article{GHM03,
  author       = {Sylvain Gravier and
                  Ch{\'{\i}}nh T. Ho{\`{a}}ng and
                  Fr{\'{e}}d{\'{e}}ric Maffray},
  title        = {Coloring the hypergraph of maximal cliques of a graph with no long
                  path},
  journal      = {Discret. Math.},
  volume       = {272},
  number       = {2-3},
  pages        = {285--290},
  year         = {2003},
  url          = {https://doi.org/10.1016/S0012-365X(03)00197-3},
  doi          = {10.1016/S0012-365X(03)00197-3},
  bibsource    = {dblp computer science bibliography, https://dblp.org}
}

@article{Gy87,
author = {Andr\'as Gy\'arf\'as},
title={Problems from the world surrounding perfect graphs},
journal={Zastosowania Matematyki Applicationes Mathematicae},
volume={XIX},
year={1987},
pages={413--441}
}

@article{Br02,
  author       = {Stephan Brandt},
  title        = {Triangle-free graphs and forbidden subgraphs},
  journal      = {Discret. Appl. Math.},
  volume       = {120},
  number       = {1-3},
  pages        = {25--33},
  year         = {2002},
  url          = {https://doi.org/10.1016/S0166-218X(01)00277-3},
  doi          = {10.1016/S0166-218X(01)00277-3},
  bibsource    = {dblp computer science bibliography, https://dblp.org}
}

@article{DEJPP25,
author = {Konrad K. Dabrowski and Tala Eagling-Vose and Matthew Johnson and Giacomo Paesani and Dani\"el Paulusma},
title= {Finding $d$-cuts in probe ${H}$-free graphs}, 
journal ={Proc. FCT 2025,  Lecture Notes in Computer Science}, 
volume={to appear},
year={2025},
pages={}
}

@article{CK23,
  author       = {Arnab Char and
                  T. Karthick},
  title        = {Improved bounds on the chromatic number of $({P}_5,\mbox{flag})$-free
                  graphs},
  journal      = {Discret. Math.},
  volume       = {346},
  number       = {9},
  pages        = {113501},
  year         = {2023},
  url          = {https://doi.org/10.1016/j.disc.2023.113501},
  doi          = {10.1016/J.DISC.2023.113501},
  bibsource    = {dblp computer science bibliography, https://dblp.org}
}

@article{DXX22,
  author       = {Wei Dong and
                  Baogang Xu and
                  Yian Xu},
  title        = {On the chromatic number of some ${P}_5$-free graphs},
  journal      = {Discret. Math.},
  volume       = {345},
  number       = {10},
  pages        = {113004},
  year         = {2022},
  url          = {https://doi.org/10.1016/j.disc.2022.113004},
  doi          = {10.1016/J.DISC.2022.113004},
  bibsource    = {dblp computer science bibliography, https://dblp.org}
}

@article{DXX23,
  author       = {Wei Dong and
                  Baogang Xu and
                  Yian Xu},
  title        = {A Tight Linear Bound to the Chromatic Number of $({P}_5,{K}_1+{K}_3)$-Free
  Graphs},
  journal      = {Graphs Comb.},
  volume       = {39},
  number       = {3},
  pages        = {43},
  year         = {2023},
  url          = {https://doi.org/10.1007/s00373-023-02642-y},
  doi          = {10.1007/S00373-023-02642-Y},
  bibsource    = {dblp computer science bibliography, https://dblp.org}
}

@article{GH19,
  author       = {Serge Gaspers and
                  Shenwei Huang},
  title        = {$(2{P}_2,{K}_4)$-Free Graphs are $4$-Colorable},
  journal      = {{SIAM} J. Discret. Math.},
  volume       = {33},
  number       = {2},
  pages        = {1095--1120},
  year         = {2019},
  url          = {https://doi.org/10.1137/18M1205832},
  doi          = {10.1137/18M1205832},
  bibsource    = {dblp computer science bibliography, https://dblp.org}
}

@article{CK25,
  author       = {Arnab Char and
                  T. Karthick},
  title        = {{\(\chi\)}-boundedness and related problems on graphs without long
                  induced paths: {A} survey},
  journal      = {Discret. Appl. Math.},
  volume       = {364},
  pages        = {99--119},
  year         = {2025},
  url          = {https://doi.org/10.1016/j.dam.2024.12.014},
  doi          = {10.1016/J.DAM.2024.12.014},
  bibsource    = {dblp computer science bibliography, https://dblp.org}
}

@article{BDFJP19,
  author       = {Marthe Bonamy and
                  Konrad K. Dabrowski and
                  Carl Feghali and
                  Matthew Johnson and
                  Dani{\"{e}}l Paulusma},
  title        = {Independent {F}eedback {V}ertex {S}et for ${P}_5$-free Graphs},
  journal      = {Algorithmica},
  volume       = {81},
  number       = {4},
  pages        = {1342--1369},
  year         = {2019},
  url          = {https://doi.org/10.1007/s00453-018-0474-x},
  doi          = {10.1007/S00453-018-0474-X},
  bibsource    = {dblp computer science bibliography, https://dblp.org}
}

@article{ACPRS24,
  author       = {Tara Abrishami and
                  Maria Chudnovsky and
                  Marcin Pilipczuk and
                  Pawel Rzazewski and
                  Paul D. Seymour},
  title        = {Induced Subgraphs of Bounded Treewidth and the Container Method},
  journal      = {{SIAM} J. Comput.},
  volume       = {53},
  number       = {3},
  pages        = {624--647},
  year         = {2024},
  url          = {https://doi.org/10.1137/20m1383732},
  doi          = {10.1137/20M1383732},
  bibsource    = {dblp computer science bibliography, https://dblp.org}
}

@article{Fe23,
  author       = {Carl Feghali},
  title        = {A note on matching-cut in ${P}_t$-free graphs},
  journal      = {Inf. Process. Lett.},
  volume       = {179},
  pages        = {106294},
  year         = {2023},
  url          = {https://doi.org/10.1016/j.ipl.2022.106294},
  doi          = {10.1016/J.IPL.2022.106294},
  bibsource    = {dblp computer science bibliography, https://dblp.org}
}

@article{JKMM,
title={Constricting the Computational Complexity Gap of the $4$-{C}oloring Problem in $({P}_t,{C}_3)$-free Graphs},
author={Justyna Jaworska and Bartłomiej Kielak and Tomáš Masařík and Jana Masaříková},
journal      = {CoRR},
volume       = {abs/2509.02423},
year         = {2025}
}

@article{WS01,
  author       = {Gerhard J. Woeginger and
                  Jir{\'{\i}} Sgall},
  title        = {The complexity of coloring graphs without long induced paths},
  journal      = {Acta Cybern.},
  volume       = {15},
  number       = {1},
  pages        = {107--117},
  year         = {2001},
  url          = {https://cyber.bibl.u-szeged.hu/index.php/actcybern/article/view/3566},
  bibsource    = {dblp computer science bibliography, https://dblp.org}
}

@article{EHK98,
  author       = {Thomas Emden{-}Weinert and
                  Stefan Hougardy and
                  Bernd Kreuter},
  title        = {Uniquely Colourable Graphs and the Hardness of Colouring Graphs of
                  Large Girth},
  journal      = {Comb. Probab. Comput.},
  volume       = {7},
  number       = {4},
  pages        = {375--386},
  year         = {1998},
  url          = {http://journals.cambridge.org/action/displayAbstract?aid=46667},
  bibsource    = {dblp computer science bibliography, https://dblp.org}
}

@inproceedings{BCM15,
  author       = {Johann Brault{-}Baron and
                  Florent Capelli and
                  Stefan Mengel},
  editor       = {Ernst W. Mayr and
                  Nicolas Ollinger},
  title        = {Understanding Model Counting for beta-acyclic {CNF}-formulas},
  booktitle    = {32nd International Symposium on Theoretical Aspects of Computer Science,
                  {STACS} 2015, March 4-7, 2015, Garching, Germany},
  series       = {LIPIcs},
  volume       = {30},
  pages        = {143--156},
  publisher    = {Schloss Dagstuhl - Leibniz-Zentrum f{\"{u}}r Informatik},
  year         = {2015},
  url          = {https://doi.org/10.4230/LIPIcs.STACS.2015.143},
  doi          = {10.4230/LIPICS.STACS.2015.143},
  bibsource    = {dblp computer science bibliography, https://dblp.org}
}

@article{DLRR12,
  author       = {Konrad K. Dabrowski and
                  Vadim V. Lozin and
                  Rajiv Raman and
                  Bernard Ries},
  title        = {Colouring vertices of triangle-free graphs without forests},
  journal      = {Discret. Math.},
  volume       = {312},
  number       = {7},
  pages        = {1372--1385},
  year         = {2012},
  url          = {https://doi.org/10.1016/j.disc.2011.12.012},
  doi          = {10.1016/J.DISC.2011.12.012},
  bibsource    = {dblp computer science bibliography, https://dblp.org}
}

@article{ALLRSS,
  author       = {Akanksha Agrawal and
                  Paloma T. Lima and
                  Daniel Lokshtanov and
                  Pawe{\l} Rz{\k{a}}\.{z}ewski and
                  Saket Saurabh and
                  Roohani Sharma},
  title        = {Odd Cycle Transversal on ${P}_5$-free Graphs in Polynomial
                  Time},
  journal      = {ACM Transactions on Algorithms},
  volume       = {21},
  year         = {2025},
pages = {16:1--14}
}

@article{RST02,
  author       = {Bert Randerath and
                  Ingo Schiermeyer and
                  Meike Tewes},
  title        = {Three-colourability and forbidden subgraphs. {I}{I}: {P}olynomial algorithms},
  journal      = {Discret. Math.},
  volume       = {251},
  number       = {1-3},
  pages        = {137--153},
  year         = {2002},
  url          = {https://doi.org/10.1016/S0012-365X(01)00335-1},
  doi          = {10.1016/S0012-365X(01)00335-1},
  bibsource    = {dblp computer science bibliography, https://dblp.org}
}

@inproceedings{LVV14,
  author       = {Daniel Lokshtanov and
                  Martin Vatshelle and
                  Yngve Villanger},
  editor       = {Chandra Chekuri},
  title        = {Independent Set in ${P}_5$-Free Graphs in Polynomial
                  Time},
  booktitle    = {Proceedings of the Twenty-Fifth Annual {ACM-SIAM} Symposium on Discrete
                  Algorithms, {SODA} 2014, Portland, Oregon, USA, January 5-7, 2014},
  pages        = {570--581},
  publisher    = {{SIAM}},
  year         = {2014},
  url          = {https://doi.org/10.1137/1.9781611973402.43},
  doi          = {10.1137/1.9781611973402.43},
  bibsource    = {dblp computer science bibliography, https://dblp.org}
}

@article{LRSSZ,
  author       = {Daniel Lokshtanov and
                  Pawe{\l} Rz{\k{a}}\.{z}ewski and
                  Saket Saurabh and
                  Roohani Sharma and
                  Meirav Zehavi},
  title        = {Maximum Partial List ${H}$-Coloring on ${P}_5$-free graphs
                  in polynomial time},
  journal      = {CoRR},
  volume       = {abs/2410.21569},
  year         = {2024},
  url          = {https://doi.org/10.48550/arXiv.2410.21569},
  doi          = {10.48550/ARXIV.2410.21569},
  eprinttype    = {arXiv},
  eprint       = {2410.21569},
  bibsource    = {dblp computer science bibliography, https://dblp.org}
}

@article{BTV13,
  author       = {Binh{-}Minh Bui{-}Xuan and
                  Jan Arne Telle and
                  Martin Vatshelle},
  title        = {Fast dynamic programming for locally checkable vertex subset and vertex
                  partitioning problems},
  journal      = {Theor. Comput. Sci.},
  volume       = {511},
  pages        = {66--76},
  year         = {2013}
}

@article{HKLSS10,
  author       = {Ch{\'{\i}}nh T. Ho{\`{a}}ng and
                  Marcin Kami\'nski and
                  Vadim V. Lozin and
                  Joe Sawada and
                  Xiao Shu},
  title        = {Deciding $k$-{C}olorability of ${P}_5$-Free Graphs
                  in Polynomial Time},
  journal      = {Algorithmica},
  volume       = {57},
  number       = {1},
  pages        = {74--81},
  year         = {2010},
  url          = {https://doi.org/10.1007/s00453-008-9197-8},
  doi          = {10.1007/S00453-008-9197-8},
  bibsource    = {dblp computer science bibliography, https://dblp.org}
}

@article{KMMNPS20,
  author       = {Tereza Klimosov{\'{a}} and
                  Josef Mal{\'{\i}}k and
                  Tom{\'{a}}s Masar{\'{\i}}k and
                  Jana Novotn{\'{a}} and
                  Dani{\"{e}}l Paulusma and
                  Veronika Sl{\'{i}}vov{\'{a}}},
  title        = {Colouring $({P}_r+{P}_s)$-Free Graphs},
  journal      = {Algorithmica},
  volume       = {82},
  number       = {7},
  pages        = {1833--1858},
  year         = {2020},
  url          = {https://doi.org/10.1007/s00453-020-00675-w},
  doi          = {10.1007/S00453-020-00675-W},
  bibsource    = {dblp computer science bibliography, https://dblp.org}
}

@article{HLS22,
  author       = {Sepehr Hajebi and
                  Yanjia Li and
                  Sophie Spirkl},
  title        = {Complexity Dichotomy for {L}ist-$5$-{C}oloring with a Forbidden Induced
                  Subgraph},
  journal      = {{SIAM} J. Discret. Math.},
  volume       = {36},
  number       = {3},
  pages        = {2004--2027},
  year         = {2022},
  url          = {https://doi.org/10.1137/21m1443352},
  doi          = {10.1137/21M1443352},
  bibsource    = {dblp computer science bibliography, https://dblp.org}
}

@article{CHS24,
  author       = {Maria Chudnovsky and
                  Sepehr Hajebi and
                  Sophie Spirkl},
  title        = {List-$k$-Coloring ${H}$-Free Graphs for All $k\geq 4$},
  journal      = {Comb.},
  volume       = {44},
  number       = {5},
  pages        = {1063--1068},
  year         = {2024},
  url          = {https://doi.org/10.1007/s00493-024-00106-2},
  doi          = {10.1007/S00493-024-00106-2},
  bibsource    = {dblp computer science bibliography, https://dblp.org}
}

@inproceedings{PPR21,
  author       = {Marcin Pilipczuk and
                  Michal Pilipczuk and
                  Pawe{\l} Rz{\k{a}}\.{z}ewski},
  editor       = {Hung Viet Le and
                  Valerie King},
  title        = {Quasi-polynomial-time algorithm for Independent Set in ${P}_t$-free
                  graphs via shrinking the space of induced paths},
  booktitle    = {4th Symposium on Simplicity in Algorithms, {SOSA} 2021, Virtual Conference,
                  January 11-12, 2021},
  pages        = {204--209},
  publisher    = {{SIAM}},
  year         = {2021},
  url          = {https://doi.org/10.1137/1.9781611976496.23},
  doi          = {10.1137/1.9781611976496.23},
  bibsource    = {dblp computer science bibliography, https://dblp.org}
}

@article{CHSZ21,
  author       = {Maria Chudnovsky and
                  Shenwei Huang and
                  Sophie Spirkl and
                  Mingxian Zhong},
  title        = {List $3$-Coloring Graphs with No Induced ${P}_6+r{P}_3$},
  journal      = {Algorithmica},
  volume       = {83},
  number       = {1},
  pages        = {216--251},
  year         = {2021},
  url          = {https://doi.org/10.1007/s00453-020-00754-y},
  doi          = {10.1007/S00453-020-00754-Y},
  bibsource    = {dblp computer science bibliography, https://dblp.org}
}

@article{GJPS17,
  author       = {Petr A. Golovach and
                  Matthew Johnson and
                  Dani{\"{e}}l Paulusma and
                  Jian Song},
  title        = {A Survey on the Computational Complexity of Coloring Graphs with Forbidden
                  Subgraphs},
  journal      = {J. Graph Theory},
  volume       = {84},
  number       = {4},
  pages        = {331--363},
  year         = {2017},
  url          = {https://doi.org/10.1002/jgt.22028},
  doi          = {10.1002/JGT.22028},
  bibsource    = {dblp computer science bibliography, https://dblp.org}
}

@article{Hu16,
  author       = {Shenwei Huang},
  title        = {Improved complexity results on $k$-coloring ${P}_t$-free
                  graphs},
  journal      = {Eur. J. Comb.},
  volume       = {51},
  pages        = {336--346},
  year         = {2016},
  url          = {https://doi.org/10.1016/j.ejc.2015.06.005},
  doi          = {10.1016/J.EJC.2015.06.005},
  bibsource    = {dblp computer science bibliography, https://dblp.org}
}

@article{BBMP24,
  author       = {Flavia Bonomo{-}Braberman and
                  Nick Brettell and
                  Andrea Munaro and
                  Dani{\"{e}}l Paulusma},
  title        = {Solving problems on generalized convex graphs via mim-width},
  journal      = {J. Comput. Syst. Sci.},
  volume       = {140},
  pages        = {103493},
  year         = {2024},
  url          = {https://doi.org/10.1016/j.jcss.2023.103493},
  doi          = {10.1016/J.JCSS.2023.103493},
  bibsource    = {dblp computer science bibliography, https://dblp.org}
}

@article{BCMSSZ18,
  author       = {Flavia Bonomo and
                  Maria Chudnovsky and
                  Peter Maceli and
                  Oliver Schaudt and
                  Maya Stein and
                  Mingxian Zhong},
  title        = {Three-Coloring and List Three-Coloring of Graphs Without Induced Paths
                  on Seven Vertices},
  journal      = {Comb.},
  volume       = {38},
  number       = {4},
  pages        = {779--801},
  year         = {2018},
  url          = {https://doi.org/10.1007/s00493-017-3553-8},
  doi          = {10.1007/S00493-017-3553-8},
  bibsource    = {dblp computer science bibliography, https://dblp.org}
}

@article{LG83,
  author       = {Daniel Leven and
                  Zvi Galil},
  title        = {{N}{P} Completeness of Finding the Chromatic Index of Regular Graphs},
  journal      = {J. Algorithms},
  volume       = {4},
  number       = {1},
  pages        = {35--44},
  year         = {1983},
  url          = {https://doi.org/10.1016/0196-6774(83)90032-9},
  doi          = {10.1016/0196-6774(83)90032-9},
  bibsource    = {dblp computer science bibliography, https://dblp.org}
}

@article{aspvall1979lineartime,
	author       = {Bengt Aspvall and
	Michael F. Plass and
	Robert E. Tarjan},
	title        = {A Linear-Time Algorithm for Testing the Truth of Certain Quantified
	Boolean Formulas},
	journal      = {Information Processing Letters},
	volume       = {8},
	number       = {3},
	pages        = {121--123},
	year         = {1979},
	doi          = {10.1016/0020-0190(79)90002-4},
	bibsource    = {dblp computer science bibliography, https://dblp.org}
}

@article{camby2016newchar,
  author       = {Eglantine Camby and
                  Oliver Schaudt},
  title        = {A New Characterization of {$P_k$}-Free Graphs},
  journal      = {Algorithmica},
  volume       = {75},
  number       = {1},
  pages        = {205--217},
  year         = {2016},
  url          = {https://doi.org/10.1007/s00453-015-9989-6},
  doi          = {10.1007/S00453-015-9989-6},
  bibsource    = {dblp computer science bibliography, https://dblp.org}
}

@article{bodlaender1994scheduling,
  author       = {Hans L. Bodlaender and
                  Klaus Jansen and
                  Gerhard J. Woeginger},
  title        = {Scheduling with Incompatible Jobs},
  journal      = {Discret. Appl. Math.},
  volume       = {55},
  number       = {3},
  pages        = {219--232},
  year         = {1994},
  url          = {https://doi.org/10.1016/0166-218X(94)90009-4},
  doi          = {10.1016/0166-218X(94)90009-4},
  bibsource    = {dblp computer science bibliography, https://dblp.org}
}

@book{garey1979computers,
  author       = {M. R. Garey and
                  David S. Johnson},
  title        = {Computers and Intractability: {A} Guide to the Theory of NP-Completeness},
  publisher    = {W. H. Freeman},
  year         = {1979},
  isbn         = {0-7167-1044-7},
  bibsource    = {dblp computer science bibliography, https://dblp.org}
}

@inproceedings{kral2001complexity,
	author       = {Daniel Kr{\'{a}}l and
	Jan Kratochv{\'{\i}}l and
	Zsolt Tuza and
	Gerhard J. Woeginger},
	editor       = {Andreas Brandst{\"{a}}dt and
	Van Bang Le},
	title        = {Complexity of Coloring Graphs without Forbidden Induced Subgraphs},
	booktitle    = {Graph-Theoretic Concepts in Computer Science, 27th International Workshop,
	{WG} 2001, Boltenhagen, Germany, June 14-16, 2001, Proceedings},
	series       = {Lecture Notes in Computer Science},
	volume       = {2204},
	pages        = {254--262},
	publisher    = {Springer},
	year         = {2001},
	url          = {https://doi.org/10.1007/3-540-45477-2\_23},
	doi          = {10.1007/3-540-45477-2\_23},
	bibsource    = {dblp computer science bibliography, https://dblp.org}
}

@article{holyer1981thenpcompleteness,
	author       = {Ian Holyer},
	title        = {The {N}{P}-Completeness of {E}dge-{C}oloring},
	journal      = {{SIAM} J. Comput.},
	volume       = {10},
	number       = {4},
	pages        = {718--720},
	year         = {1981},
	url          = {https://doi.org/10.1137/0210055},
	doi          = {10.1137/0210055},
	bibsource    = {dblp computer science bibliography, https://dblp.org}
}

@article{broersma2012updating,
	author       = {Hajo Broersma and
	Petr A. Golovach and
	Dani{\"{e}}l Paulusma and
	Jian Song},
	title        = {Updating the complexity status of coloring graphs without a fixed
	induced linear forest},
	journal      = {Theor. Comput. Sci.},
	volume       = {414},
	number       = {1},
	pages        = {9--19},
	year         = {2012},
	url          = {https://doi.org/10.1016/j.tcs.2011.10.005},
	doi          = {10.1016/J.TCS.2011.10.005},
	bibsource    = {dblp computer science bibliography, https://dblp.org}
}

@article{zhang1994analgorithm,
	author       = {Peisen Zhang and
	Eric A. Schon and
	Stuart G. Fischer and
	Eftihia Cayanis and
	Janie Weiss and
	Susan Kistler and
	Philip E. Bourne},
	title        = {An algorithm based on graph theory for the assembly of contigs in
	physical mapping of {DNA}},
	journal      = {Comput. Appl. Biosci.},
	volume       = {10},
	number       = {3},
	pages        = {309--317},
	year         = {1994},
	url          = {https://doi.org/10.1093/bioinformatics/10.3.309},
	doi          = {10.1093/BIOINFORMATICS/10.3.309},
	bibsource    = {dblp computer science bibliography, https://dblp.org}
}

@inproceedings{chang2005ontherecognition,
	author       = {Maw{-}Shang Chang and
	Ton Kloks and
	Dieter Kratsch and
	Jiping Liu and
	Sheng{-}Lung Peng},
	editor       = {Lusheng Wang},
	title        = {On the Recognition of Probe Graphs of Some Self-Complementary Classes
	of Perfect Graphs},
	booktitle    = {Computing and Combinatorics, 11th Annual International Conference,
	{COCOON} 2005, Kunming, China, August 16-29, 2005, Proceedings},
	series       = {Lecture Notes in Computer Science},
	volume       = {3595},
	pages        = {808--817},
	publisher    = {Springer},
	year         = {2005},
	url          = {https://doi.org/10.1007/11533719\_82},
	doi          = {10.1007/11533719\_82},
	bibsource    = {dblp computer science bibliography, https://dblp.org}
}

@article{golumbic2004chordalprobe,
	author       = {Martin Charles Golumbic and
	Marina Lipshteyn},
	title        = {Chordal probe graphs},
	journal      = {Discret. Appl. Math.},
	volume       = {143},
	number       = {1-3},
	pages        = {221--237},
	year         = {2004},
	url          = {https://doi.org/10.1016/j.dam.2003.12.009},
	doi          = {10.1016/J.DAM.2003.12.009},
	bibsource    = {dblp computer science bibliography, https://dblp.org}
}

@article{BrettellOPPRL25,
  author       = {Nick Brettell and
                  Jelle J. Oostveen and
                  Sukanya Pandey and
                  Dani{\"{e}}l Paulusma and
                  Johannes Rauch and
                  Erik Jan van Leeuwen},
  title        = {Computing subset vertex covers in ${H}$-free graphs},
  journal      = {Theor. Comput. Sci.},
  volume       = {1032},
  pages        = {115088},
  year         = {2025},
  url          = {https://doi.org/10.1016/j.tcs.2025.115088},
  doi          = {10.1016/J.TCS.2025.115088},
  bibsource    = {dblp computer science bibliography, https://dblp.org}
}

@article{brettel2022listkcolouring,
	author       = {Nick Brettell and
	Jake Horsfield and
	Andrea Munaro and
	Dani{\"{e}}l Paulusma},
	title        = {List $k$-{C}olouring {$P_t$}-free graphs:
	{A} Mim-width perspective},
	journal      = {Inf. Process. Lett.},
	volume       = {173},
	pages        = {106168},
	year         = {2022},
	url          = {https://doi.org/10.1016/j.ipl.2021.106168},
	doi          = {10.1016/J.IPL.2021.106168},
	bibsource    = {dblp computer science bibliography, https://dblp.org}
}

@article{brettel2022boundingmimwidth,
	author       = {Nick Brettell and
	Jake Horsfield and
	Andrea Munaro and
	Giacomo Paesani and
	Dani{\"{e}}l Paulusma},
	title        = {Bounding the mim-width of hereditary graph classes},
	journal      = {J. Graph Theory},
	volume       = {99},
	number       = {1},
	pages        = {117--151},
	year         = {2022},
	url          = {https://doi.org/10.1002/jgt.22730},
	doi          = {10.1002/JGT.22730},
	bibsource    = {dblp computer science bibliography, https://dblp.org}
}

@article{blanche2019hereditary,
	author       = {Alexandre Blanch{\'{e}} and
	Konrad K. Dabrowski and
	Matthew Johnson and
	Dani{\"{e}}l Paulusma},
	title        = {Hereditary graph classes: When the complexities of coloring and clique
	cover coincide},
	journal      = {J. Graph Theory},
	volume       = {91},
	number       = {3},
	pages        = {267--289},
	year         = {2019},
	url          = {https://doi.org/10.1002/jgt.22431},
	doi          = {10.1002/JGT.22431},
	bibsource    = {dblp computer science bibliography, https://dblp.org}
}

@article{courcelle2000upperbounds,
  author       = {Bruno Courcelle and
                  Stephan Olariu},
  title        = {Upper bounds to the clique width of graphs},
  journal      = {Discret. Appl. Math.},
  volume       = {101},
  number       = {1-3},
  pages        = {77--114},
  year         = {2000},
  url          = {https://doi.org/10.1016/S0166-218X(99)00184-5},
  doi          = {10.1016/S0166-218X(99)00184-5},
  bibsource    = {dblp computer science bibliography, https://dblp.org}
}

@inproceedings{EspelageGW01,
  author       = {Wolfgang Espelage and
                  Frank Gurski and
                  Egon Wanke},
  editor       = {Andreas Brandst{\"{a}}dt and
                  Van Bang Le},
  title        = {How to Solve {NP}-hard Graph Problems on Clique-Width Bounded Graphs
                  in Polynomial Time},
  booktitle    = {Graph-Theoretic Concepts in Computer Science, 27th International Workshop,
                  {WG} 2001, Boltenhagen, Germany, June 14-16, 2001, Proceedings},
  series       = {Lecture Notes in Computer Science},
  volume       = {2204},
  pages        = {117--128},
  publisher    = {Springer},
  year         = {2001},
  url          = {https://doi.org/10.1007/3-540-45477-2\_12},
  doi          = {10.1007/3-540-45477-2\_12},
  bibsource    = {dblp computer science bibliography, https://dblp.org}
}

@article{KoblerR03,
  author       = {Daniel Kobler and
                  Udi Rotics},
  title        = {Edge dominating set and colorings on graphs with fixed clique-width},
  journal      = {Discret. Appl. Math.},
  volume       = {126},
  number       = {2-3},
  pages        = {197--221},
  year         = {2003},
  url          = {https://doi.org/10.1016/S0166-218X(02)00198-1},
  doi          = {10.1016/S0166-218X(02)00198-1},
  bibsource    = {dblp computer science bibliography, https://dblp.org}
}

@article{HlinenyO08,
  author       = {Petr Hlinen{\'{y}} and
                  Sang{-}il Oum},
  title        = {Finding Branch-Decompositions and Rank-Decompositions},
  journal      = {{SIAM} J. Comput.},
  volume       = {38},
  number       = {3},
  pages        = {1012--1032},
  year         = {2008},
  url          = {https://doi.org/10.1137/070685920},
  doi          = {10.1137/070685920},
  bibsource    = {dblp computer science bibliography, https://dblp.org}
}

@article{Edwards86,
  author       = {Keith Edwards},
  title        = {The Complexity of Colouring Problems on Dense Graphs},
  journal      = {Theor. Comput. Sci.},
  volume       = {43},
  pages        = {337--343},
  year         = {1986},
  url          = {https://doi.org/10.1016/0304-3975(86)90184-2},
  doi          = {10.1016/0304-3975(86)90184-2},
  bibsource    = {dblp computer science bibliography, https://dblp.org}
}

@article{chudnovsky2024fourcolouring1,
  author       = {Maria Chudnovsky and
                  Sophie Spirkl and
                  Mingxian Zhong},
  title        = {Four-Coloring ${P}_6$-Free Graphs.
                  {I}. Extending an Excellent Precoloring},
  journal      = {{SIAM} J. Comput.},
  volume       = {53},
  number       = {1},
  pages        = {111--145},
  year         = {2024},
  url          = {https://doi.org/10.1137/18m1234837},
  doi          = {10.1137/18M1234837},
  bibsource    = {dblp computer science bibliography, https://dblp.org}
}

@article{chudnovsky2024fourcolouring2,
  author       = {Maria Chudnovsky and
                  Sophie Spirkl and
                  Mingxian Zhong},
  title        = {Four-Coloring ${P}_6$-Free
                  Graphs. {II.} Finding an Excellent Precoloring},
  journal      = {{SIAM} J. Comput.},
  volume       = {53},
  number       = {1},
  pages        = {146--187},
  year         = {2024},
  url          = {https://doi.org/10.1137/18m1234849},
  doi          = {10.1137/18M1234849},
  bibsource    = {dblp computer science bibliography, https://dblp.org}
}

@article{chandler2009onprobepermutationgraphs,
  author       = {David B. Chandler and
                  Maw{-}Shang Chang and
                  Ton Kloks and
                  Jiping Liu and
                  Sheng{-}Lung Peng},
  title        = {On probe permutation graphs},
  journal      = {Discret. Appl. Math.},
  volume       = {157},
  number       = {12},
  pages        = {2611--2619},
  year         = {2009},
  url          = {https://doi.org/10.1016/j.dam.2008.08.017},
  doi          = {10.1016/J.DAM.2008.08.017},
  bibsource    = {dblp computer science bibliography, https://dblp.org}
}

@article{bayer2009probethresholdprobetriviallyperfect,
  author       = {Daniel Bayer and
                  Van Bang Le and
                  H. N. de Ridder},
  title        = {Probe threshold and probe trivially perfect graphs},
  journal      = {Theor. Comput. Sci.},
  volume       = {410},
  number       = {47-49},
  pages        = {4812--4822},
  year         = {2009},
  url          = {https://doi.org/10.1016/j.tcs.2009.06.029},
  doi          = {10.1016/J.TCS.2009.06.029},
  bibsource    = {dblp computer science bibliography, https://dblp.org}
}

@article{golumbic2011chainprobegraphs,
  author       = {Martin Charles Golumbic and
                  Fr{\'{e}}d{\'{e}}ric Maffray and
                  Gr{\'{e}}gory Morel},
  title        = {A characterization of chain probe graphs},
  journal      = {Ann. Oper. Res.},
  volume       = {188},
  number       = {1},
  pages        = {175--183},
  year         = {2011},
  url          = {https://doi.org/10.1007/s10479-009-0584-6},
  doi          = {10.1007/S10479-009-0584-6},
  bibsource    = {dblp computer science bibliography, https://dblp.org}
}

@article{berry2007recognizingchordalprobegraphs,
  author       = {Anne Berry and
                  Martin Charles Golumbic and
                  Marina Lipshteyn},
  title        = {Recognizing Chordal Probe Graphs and Cycle-Bicolorable Graphs},
  journal      = {{SIAM} J. Discret. Math.},
  volume       = {21},
  number       = {3},
  pages        = {573--591},
  year         = {2007},
  url          = {https://doi.org/10.1137/050637091},
  doi          = {10.1137/050637091},
  bibsource    = {dblp computer science bibliography, https://dblp.org}
}

@article{le2015recognizingprobeblockgraphs,
  author       = {Van Bang Le and
                  Sheng{-}Lung Peng},
  title        = {Characterizing and recognizing probe block graphs},
  journal      = {Theor. Comput. Sci.},
  volume       = {568},
  pages        = {97--102},
  year         = {2015},
  url          = {https://doi.org/10.1016/j.tcs.2014.12.014},
  doi          = {10.1016/J.TCS.2014.12.014},
  bibsource    = {dblp computer science bibliography, https://dblp.org}
}

@article{groetschel1981ellipsoidmethodcombinatorialoptimization,
  author       = {Martin Gr{\"{o}}tschel and
                  L{\'{a}}szl{\'{o}} Lov{\'{a}}sz and
                  Alexander Schrijver},
  title        = {The ellipsoid method and its consequences in combinatorial optimization},
  journal      = {Comb.},
  volume       = {1},
  number       = {2},
  pages        = {169--197},
  year         = {1981},
  url          = {https://doi.org/10.1007/BF02579273},
  doi          = {10.1007/BF02579273},
  bibsource    = {dblp computer science bibliography, https://dblp.org}
}

@article{groetschel1984ellipsoidmethodcombinatorialoptimization,
  author       = {Martin Gr{\"{o}}tschel and
                  L{\'{a}}szl{\'{o}} Lov{\'{a}}sz and
                  Alexander Schrijver},
  title        = {Corrigendum to our paper "The ellipsoid method and its consequences
                  in combinatorial optimization"},
  journal      = {Comb.},
  volume       = {4},
  number       = {4},
  pages        = {291--295},
  year         = {1984},
  url          = {https://doi.org/10.1007/BF02579139},
  doi          = {10.1007/BF02579139},
  bibsource    = {dblp computer science bibliography, https://dblp.org}
}

@manuscript{chandler2012probegraphclasses,
  author       = {David B. Chandler and Maw-Shang Chang and Ton Kloks and Jiping Liu and Sheng-Lung Peng},
  title        = {Probe Graph Classes},
  year         = {2012},
  note = {unpublished},
  url = {https://citeseerx.ist.psu.edu/document?repid=rep1&type=pdf&doi=a0f280169bb1c7ec80fcfd142bf7982448ed4e3b}
}

@article{dabrowski2016cliquewidth,
  author       = {Konrad K. Dabrowski and
                  Dani{\"{e}}l Paulusma},
  title        = {Clique-Width of Graph Classes Defined by Two Forbidden Induced Subgraphs},
  journal      = {Comput. J.},
  volume       = {59},
  number       = {5},
  pages        = {650--666},
  year         = {2016},
  url          = {https://doi.org/10.1093/comjnl/bxv096},
  doi          = {10.1093/COMJNL/BXV096},
  bibsource    = {dblp computer science bibliography, https://dblp.org}
}

@article{CourcelleER93,
  author       = {Bruno Courcelle and
                  Joost Engelfriet and
                  Grzegorz Rozenberg},
  title        = {Handle-Rewriting Hypergraph Grammars},
  journal      = {J. Comput. Syst. Sci.},
  volume       = {46},
  number       = {2},
  pages        = {218--270},
  year         = {1993},
  url          = {https://doi.org/10.1016/0022-0000(93)90004-G},
  doi          = {10.1016/0022-0000(93)90004-G},
  bibsource    = {dblp computer science bibliography, https://dblp.org}
}

\end{document}